\documentclass[12pt]{article}
\usepackage{setspace}
\usepackage{natbib}
\usepackage{graphicx}
\usepackage{lscape}
\usepackage{pdflscape}
\usepackage{afterpage}
\usepackage{pdfpages}
\usepackage{float}
\usepackage{booktabs}
\usepackage[utf8]{inputenc}
\usepackage{indentfirst}
\usepackage{arydshln}
\usepackage{arydshln}
\usepackage{tikz}
\usepackage{lipsum}
\usepackage{pgffor}
\usepackage{dsfont}
\usepackage{amsfonts}
\usepackage{amsmath}
\usepackage{xfrac}
\usepackage{amssymb}
\usepackage{graphicx}
\usepackage{url}				
\usepackage{subfig} 

\newcommand\norm[1]{\left\lVert#1\right\rVert}

\newtheorem{innercustomgeneric}{\customgenericname}
\providecommand{\customgenericname}{}
\newcommand{\newcustomtheorem}[2]{%
  \newenvironment{#1}[1]
  {%
   \renewcommand\customgenericname{#2}%
   \renewcommand\theinnercustomgeneric{##1}%
   \innercustomgeneric
  }
  {\endinnercustomgeneric}
}

\newtheorem{lemma}{Lemma}[section]

\newtheorem{proposition}{Proposition}[section]
\newtheorem{remark}{Remark}[section]

\newenvironment{proof}[1][Proof]{\noindent \textbf{#1.} }{\  \rule{0.5em}{0.5em}}

\newtheorem{assumption}{Assumption}[section]

\newcustomtheorem{proposition_b}{Proposition}

\usepackage{hyperref} 
\hypersetup{
	colorlinks=true, breaklinks=true, bookmarks=true,bookmarksnumbered,
	urlcolor=red, linkcolor=red, citecolor=blue, 
	pdftitle={}, 
	pdfauthor={\textcopyright}, 
	pdfsubject={}, 
	pdfkeywords={}, 
	pdfcreator={pdfLaTeX}, 
	pdfproducer={LaTeX with hyperref and ClassicThesis} 
}

\usepackage[top=1in, bottom=1in, left=1in, right=1in]{geometry}

\begin{document}

	\def\spacingset#1{\renewcommand{\baselinestretch}%
		{#1}\small\normalsize} \spacingset{1}

	
\title{ 
\huge Matching Estimators with Few Treated and Many Control Observations\footnote{   The author gratefully acknowledges the comments and suggestions of Luis Alvarez, Ricardo Paes de Barros, Lucas Finamor, Sergio Firpo, Michael Jansson, Ricardo Masini, Cristine Pinto, Vitor Possebom, Pedro Sant'Anna, Azeem Shaikh, and participants of the 2017 California Econometrics Conference and of the Rio-Sao Paulo Econometrics Conference. Lucas Barros, Deivis Angeli, and Raoni de Oliveira provided outstanding research assistance. }}

\author{
Bruno Ferman\footnote{email: bruno.ferman@fgv.br; address: Sao Paulo School of Economics, FGV, Rua Itapeva no. 474, Sao Paulo - Brazil, 01332-000; telephone number: +55 11 3799-3350}  \\
Sao Paulo School of Economics - FGV \\
\\
\footnotesize
First Draft: May, 2017 \\
\footnotesize
This Draft: March, 2021
}

\date{}
\maketitle

	\newsavebox{\tablebox} \newlength{\tableboxwidth}
	

	\begin{center}

\href{https://sites.google.com/site/brunoferman/research}{Please click here for the most recent version}

\

\

\textbf{Abstract}

\end{center}

   We analyze the properties of matching estimators when there are few  treated, but many control  observations. We show that, under standard assumptions, the nearest neighbor matching estimator for the average treatment effect on the treated is asymptotically unbiased in this framework. However, when the number of treated observations is fixed, the estimator is not consistent, and it is generally not asymptotically normal. Since standard inference methods are inadequate, we propose alternative inference methods, based on the theory of randomization tests under approximate symmetry, that are  asymptotically valid in this framework. We show that these tests are valid under relatively  strong assumptions when the number of treated observations is fixed, and under weaker assumptions when the number of treated observations increases, but at a lower rate relative to the number of control observations.

\

	\noindent%
	{\it Keywords:}  matching estimators, treatment effects, hypothesis testing, randomization inference, synthetic control estimator
	
	\
	
	\noindent%
	{\it JEL Codes:} C12; C13; C21
	\vfill

	\newpage
	\spacingset{1.45} 
	

\onehalfspacing

  \section{Introduction}

Matching estimators have been widely used for the estimation of treatment effects under a conditional independence assumption (CIA).\footnote{See, for example, \cite{Imbens2004}, \cite{Imbens_Wooldrige}, and \cite{Imbens_examples} for reviews.} In many cases, matching estimators have been applied in settings where (1) the interest is on the average treatment effect for the treated (ATT), and (2) there is a large reservoir of potential controls (see \cite{Imbens_Wooldrige}). \cite{AI_2006} (henceforth, AI) study the asymptotic properties of nearest-neighbor (NN) matching estimators when the number of control observations ($N_0$) grows at a faster rate than the number of treated observations ($N_1$).
However, their asymptotic theory still depends on both the number of treated and control observations going to infinity. Therefore, reliance on such asymptotic approximation should be considered with caution when the number of treated observations is small, even if the total number of observations is large.

In this paper, we analyze the properties of  NN matching estimators when  $N_1$  is fixed, while $N_0$ goes to infinity. We first show that the NN matching estimator is asymptotically unbiased for the ATT, under standard assumptions used in the literature on estimation of treatment effects under selection on observables.\footnote{This is true whether  asymptotic unbiasedness is defined  based on the limit of the expected value of the estimator, or based on the expected value of the asymptotic distribution.} This  is consistent with the conclusions from  AI, who show that the conditional bias of the NN matching estimator can be ignored, provided that $N_0$ increases fast enough, relative to $N_1$. In their setting, the NN matching estimator is consistent and asymptotically normal. In our setting, however, the variance of the estimator does not converge to zero, and  the estimator will not generally be asymptotically normal.\footnote{Our setting is  different from the case of limited overlap considered by \cite{KT}. The problem we analyze can arise even when the overlap condition considered by AI in their Assumption 2$'$(ii) is satisfied.  The difference relative to the case considered by AI  is that $N_1$ remains fixed, so it is not possible to apply a law of large numbers and  a central limit theorem on the average of the errors of the treated observations.       } Our theory complements the theory developed by AI, providing a better approximation to settings in which  there is a  larger number of control  relative to treated observations, but $N_1$ is not large enough, so that we cannot rely on asymptotic results in which $N_1$ goes to infinity.\footnote{The finite sample properties of matching and other related estimators have been evaluated  in  simulations by, for example,  \cite{Frolich}, \cite{Busso},  \cite{Huber}, and \cite{Bodory}. In contrast to their approach, we provide theoretical and simulation results holding the number of treated observations fixed, but relying on the number of control observations going to infinity.}

The fact that the NN matching estimator is not asymptotically normal, in our setting, poses important challenges when it comes to inference. Inference based on the asymptotic distribution of the matching estimator derived by AI should not provide a good approximation when $N_1$ is very small, even if there are many control observations. The bootstrap procedure proposed by \cite{Otsu} also relies on the number of both treated and control observations going to infinity. \cite{Leung} consider a  finite population setting with limited overlap, where the probability of treatment may converge to zero for some strata. While they provide conditions in which standard inference methods remain asymptotically valid in this case, our setting with $N_1$ fixed would not satisfy their conditions. \cite{Rothe} provides robust confidence intervals for average treatment effects under limited overlap. For the case with continuous covariates, he combines his method with subclassification on the propensity score. However, with few treated and many control observations, it would not be possible to reliably estimate a propensity score. Moreover, while \cite{armstrong2017finitesample} propose confidence intervals that are asymptotically valid even when $\sqrt{n}-$inference is not possible, the conditions they consider for this result are not satisfied in our setting.\footnote{\cite{armstrong2017finitesample}  present  finite-sample results considering a setting in which errors are normal with known variance. Then they relax these conditions and consider a feasible version of their confidence intervals that is asymptotically valid. However, with $N_1$ fixed, we cannot have condition (21) in their paper being satisfied. Therefore, the results from  their Theorem 4.2 cannot be directly applied to our setting. }
  Finally, for finite samples, \cite{rosenbaum1984} and \cite{rosenbaum2002} consider permutation tests for observational studies under strong ignorability. However, these tests rely on restrictive assumptions.\footnote{\cite{rosenbaum1984} assumes that the propensity score follows a logit model, while \cite{rosenbaum2002} assumes that observations are matched in pairs such that the probability of treatment assignment is the same conditional on the pair. }  

Given the limitations of existing inference methods for the setting we analyze, we  consider alternative inference methods based on the theory of randomization tests under an approximate symmetry assumption, developed by \cite{Canay}. We focus on a test based on sign changes.  We show that, under symmetry assumptions on the errors and on the heterogeneous treatment effects, this test provides asymptotically valid hypothesis testing for the ATT when $N_0 \rightarrow \infty$,  even when $N_1$ is fixed.  When $N_1$ increases, but at a lower rate than $N_0$, we show that  this test is asymptotically valid even when we relax such symmetry conditions. Therefore, this test works with very few treated observations under relatively strong assumptions, and with a larger number of treated observations under weaker assumptions. We consider in Appendix \ref{Section_permutation} an alternative test based on permutations, which also has the property of being valid under stronger assumptions when $N_1$ is fixed, and under weaker assumptions when $N_1$ increases.

The remainder of this paper proceeds as follows. We present our theoretical setup in Section \ref{setting}.  In Section \ref{Section_unbiased}, we derive the asymptotic distribution of the NN matching estimator, and derive conditions under which it is asymptotically unbiased in a setting with fixed $N_1$.  In Section \ref{inference}, we consider an alternative inference method that is asymptotically valid when $N_0 \rightarrow \infty$, while $N_1$ remains fixed. We also consider the properties of this test when $N_1$  increases.  In Section \ref{MC}, we present  Monte Carlo (MC) simulations. In Section \ref{comparison}, we contrast the different inference procedures in light of the theoretical results presented in Section \ref{inference} and the simulations presented in Section \ref{MC}, providing guidance on which method should be chosen depending on the setting. We present in  Section \ref{empirical_application} an empirical illustration based on the ``Jovem de Futuro'' program, which provides an example in which matching estimators could be used in settings with few treated and many control observations. Concluding remarks, including a discussion on the implications of our results for other types of matching estimators and for Synthetic Control applications, are presented in Section \ref{conclusion}.

\section{Setting and Notation} \label{setting}

We are interested in estimating the effect of a binary treatment ($W$) on some outcome ($Y$). Following \cite{Rubin1973}, we define $Y(0)$ as the potential outcome under no exposure to treatment, and $Y(1)$ as the potential outcome under exposure to treatment. Therefore, the observed outcome is given by   $Y = W Y(1) + (1-W)Y(0)$. In addition to $Y$ and $W$, we also consider a continuous random vector of $k$ real-valued pretreatment variables, which we denote by $X$.\footnote{We discuss in Appendix \ref{Appendix_discrete} cases in which components of $X$ are discrete, and cases in which components of $X$ have a mixed distribution.}

We  observe $N_1$ treated observations obtained by random sampling from the distribution of $(Y,X)|W=1$ and $N_0$ untreated observations obtained by random sampling from the distribution of  $(Y,X)|W=0$. Let $\mathcal{I}_w$ denote the set of indexes for observations with $W_i=w$. 

\begin{assumption}[Sample] \label{sample}

$\{Y_i , X_i , W_i \}_{i \in \mathcal{I}_1 \cup \mathcal{I}_0}$ is a pooled sample of $N_1$ treated ($i \in \mathcal{I}_1$) and $N_0$ untreated  ($i \in \mathcal{I}_0$) observations obtained by random sampling from their respective population counterparts. Furthermore, observations in the treated and control samples are independent.

\end{assumption}

We  consider the case in which  $N_1$ is fixed, while $N_0$ goes to infinity. One possibility is that there is a large set of units that could potentially be treated, but only a finite number of them actually receive treatment. For example, in the empirical application, to be presented in Section \ref{empirical_application}, there is a large number of schools that could potentially receive the treatment, but only a small number of them actually received it. Alternatively, we can imagine that there is a large number of treated units, but we only have data from  a small sample of them. Assumption \ref{sample} is similar to Assumption 3$'$ from AI and from the first condition stated in Theorem 1 from \cite{AI_martingale}, in that the proportions of treated and control observations in the sample may not reflect their proportions in the population.

The goal is   estimating the  ATT, which we denote by 
\begin{eqnarray}
\tau \equiv  \mathbb{E} \left[ Y(1) - Y(0) | W=1 \right]. 
\end{eqnarray}

  We focus on an estimand related to the treatment effect on the treated because, given our setting with  $N_1$ finite and $N_0$ large, we would only have a small number of treated observations to serve as potential neighbors to estimate the counterfactual of the control observations, in case we wanted to estimate the average treatment effect (ATE). Likewise, AI consider the estimation of the ATT when they consider an asymptotic framework in which $N_0$ grows at a faster rate than $N_1$. In Appendix \ref{Appendix_CATT} we discuss the case in which the estimand of interest is the ATT conditional on the realization of the covariates for the treated observations, $\{X_i \}_{i \in \mathcal{I}_1}$.

Assumption \ref{sample} does not impose any restriction on how the distribution of $(Y(1),Y(0),X)$ conditional on $W = w$ depends on $w$.  The following assumption  restricts the way in which these distributions  may differ, which is a standard conditional independence assumption (CIA).

\begin{assumption}[Conditional Independence Assumption] \label{CIA}

$Y(0) \perp\!\!\!\perp W | X$.

\end{assumption}

While Assumption \ref{CIA} restricts that the conditional distribution of $Y(0)$ given $X$ is the same for both treatment and control observations, the density of $X$ conditional on $W=w \in \{0,1\}$   can potentially depend on $w$. This  is what potentially generates bias in a simple comparison of means between treated and control groups, without taking into account that these groups might have different distributions of covariates $X$. We do not need to impose conditional independence of $Y(1)$ because the focus is on the ATT, and not on the average treatment effects.

The next assumption states conditions on the distribution of the covariates. Let $\mathbb{X}_w$ be the support of $X$ conditional on $W=w$, and $f_w : \mathbb{X}_w \rightarrow \mathbb{R}_+$  be the conditional density of $X$ given $W=w$, for $w \in \{0,1\}$.

\begin{assumption}[Distribution of covariates] \label{overlap}

(i) $X \in \mathbb{R}^{k}$ is an absolutely continuous random vector, (ii)  $\mathbb{X}_1 \subseteq \mathbb{X}_0$, where  $\mathbb{X}_1$ and $\mathbb{X}_0$ are compact, (iii) $f_1$ and $f_0$ are differentiable for all points in the interior of their support,  bounded from above in $\mathbb{X}_1$, and $f_0$ is bounded from below in  $\mathbb{X}_1$, and (iv)  for all points in $x \in \mathbb{X}_1$ at least a fraction $\phi$ of any sphere around $x$ belongs to $\mathbb{X}_1$.

\end{assumption}

 This assumption guarantees that, for each $i$ in the treated group, we can find an observation $j$ in the control group with covariates $X_j$ arbitrarily close to $X_i$ when $N_0 \rightarrow \infty$. As we show in Appendix  \ref{proof_sign_changes_asympt}, Assumption \ref{overlap} implies that there is an $\eta > 0$ such that, for  all $x \in \mathbb{X}_1$, $Pr(W = 1 | X=x) < 1-\eta$.\footnote{This propensity score is defined over the distribution of $(Y,X,W)$.   }

The main identification problem arises from the fact that we observe either $Y_i(1)$ or $Y_i(0)$ for each observation $i$. If we had two observations, $i \in \mathcal{I}_1$ and $j \in \mathcal{I}_0$,  with  $X_i=X_j=x$, then, under Assumptions \ref{sample} and \ref{CIA},  $\mathbb{E}[Y_i | X_i=x] - \mathbb{E}[Y_j | X_j=x] = \mathbb{E}[Y(1) | W=1,X=x] - \mathbb{E}[Y(0) | W=0,X=x]=\mathbb{E}[Y(1) - Y(0)| X=x,W=1]$.  The main challenge is that, with a continuous random variable $X$, the probability of finding treated and control observations with exactly the same $X$ is zero. The idea of the NN matching estimator is to input the missing potential outcome of a treated observation $i \in \mathcal{I}_1$ with observations from the control group $j \in \mathcal{I}_0$ that are as close as possible in terms of covariates $X_i$. More specifically, for a distance metric $d(a,b)$ in $\mathbb{R}^k$, let $\mathcal{J}_M(i)$ be the set of $M$ nearest neighbors in the control group of observation $i \in \mathcal{I}_1$. Then the NN matching estimator is given by
\begin{eqnarray}
\hat \tau = \frac{1}{N_1} \sum_{i \in \mathcal{I}_1} \left[ Y_i - \frac{1}{M} \sum_{j \in \mathcal{J}_M(i)} Y_j \right],
\end{eqnarray}
where we consider the matching estimator with replacement. We consider the case in which $d(a,b) = [(a-b)'V(a-b)]^{1/2}$ for some positive definite matrix $V$. In Remark \ref{Mahalanobis} in the appendix we show that our results are also valid if we consider the Mahalanobis distance.

\section{Asymptotic Unbiasedness and Asymptotic Distribution } \label{Section_unbiased}

For $w \in \{ 0,1\}$, we define  $\mu(x,w) = \mathbb{E}[Y| X=x,W=w]$ and $\epsilon = Y - \mu(X,W)$. Since we are focusing on the average treatment effect on the treated, we also define  $\mu_w(x) = \mathbb{E}[Y(w)| X=x,W=1]$.\footnote{AI define $\mu_w(x) = \mathbb{E}[Y(w)| X=x]$. We use a slightly different definition because we focus on the ATT. }   Under Assumption \ref{CIA}, we have that $\mu(x,0) = \mu_0(x)$. 
Using this notation,  the ATT is given by
\begin{eqnarray}
 \tau = \mathbb{E} \left[  \mu_1(X)  - \mu_0 ( X)  | W =1 \right],
\end{eqnarray}
and the NN matching estimator is given by
\begin{eqnarray}
\hat \tau = \frac{1}{N_1}  \sum_{i \in \mathcal{I}_1} \left[ \left( \mu_1(X_i)  - \frac{1}{M} \sum_{j \in \mathcal{J}_M(i)} \mu_0 ( X_j) \right) + \left( \epsilon_i - \frac{1}{M}  \sum_{j \in \mathcal{J}_M(i)} \epsilon_j \right) \right].
\end{eqnarray}

We first show that $\hat \tau$ is an asymptotically unbiased estimator for the ATT when $N_1$ is fixed and $N_0 \rightarrow \infty$, and we derive its asymptotic distribution in this setting. We consider the following assumptions on how the distribution of $Y(0)|X = x$ changes with $x$. 

\begin{assumption}[Distribution of  $Y(0)|X = x$] \label{assumption_mu}

(a)  $\mu_0(x)$ is continuous, and (b) for any $h(y)$ continuous and bounded,  $\tilde h(x) = \mathbb{E}[h(Y(0))|X=x]$ is continuous and bounded.

\end{assumption}

Assumption  \ref{assumption_mu}(a) states that the conditional expectation of $Y(0)$ with respect to $X = x$ is continuous in $x$, which is standard in the matching literature.  The intuition behind   Assumption \ref{assumption_mu}(b) is that the conditional distribution of $Y(0)$ given $X=x$  changes ``smoothly'' with $x$. This guarantees that $Y(0)|(X= x_n)$ converges in distribution to $Y(0)|(X= x)$ if $x_n \rightarrow x$, as we show in Appendix Lemma \ref{Lemma_convergenceY}.\footnote{We use this condition to apply the Portmanteau Lemma in the proof of Appendix Lemma \ref{Lemma_convergenceY}. Other equivalent conditions could be used. }  In Appendix \ref{condition_normal}, we show that this condition is satisfied if, for example, $Y(0)|(X=x) \sim N(\theta(x),\sigma(x))$, where $\theta(x)$ and $\sigma(x)$ are continuous functions of $x$.  

For each $x \in \mathbb{X}_1$, let $\xi_x \sim (Y(1) - \mu_1(x)) | (X=x,W=1)$ and  $\eta_x  \sim (Y(0) - \mu_0(x)) | (X=x,W=0)$. Moreover, let $G(\kappa; x)$ be the CDF of $(\mu_1(x) - \mu_0(x) - \tau) +\xi_x - \frac{1}{M} \sum_{m=1}^M \eta_x^m$, where $\{\eta_x^m \}_{m=1}^M$ are iid copies of $\eta_x$, and  $(\xi_x,\eta^1_x,\cdots,\eta^M_x)$ is mutually  independent.

\begin{proposition} \label{unbiased}

(1) Under Assumptions \ref{sample}, \ref{CIA},  \ref{overlap}, and \ref{assumption_mu}(a), $\mathbb{E}[\hat \tau ] \rightarrow \tau$ when $N_0 \rightarrow \infty$ and $N_1$ is fixed.

(2) Under Assumptions \ref{sample}, \ref{CIA},  \ref{overlap}, and \ref{assumption_mu}(b), 
\begin{eqnarray} \nonumber
\hat \tau \buildrel d \over \rightarrow  \tau + \frac{1}{N_1}  \sum_{i \in \mathcal{I}_1} \kappa_i  \mbox{ when $N_0 \rightarrow \infty$ and $N_1$ is fixed}, 
\end{eqnarray}
where the CDF of $\kappa_i$ is given by $\widetilde G(\kappa) = \int_{x \in \mathbb{X}_1} G(\kappa; x)f_1(x)dx$, and $\mathbb{E}[\kappa_i]=0$.   Moreover, $\{\kappa_i\}_{i \in \mathcal{I}_1}$ is mutually independent.

\end{proposition}

Let $X^i_{(m)}$ be the covariate value of the $m$-closest match to observation $i \in \mathcal{I}_1$. The main intuition for the results in Proposition \ref{unbiased} is that, for a fixed $X_i=\bar x$, $X^i_{(m)} \buildrel p \over \rightarrow \bar x$  when $N_0 \rightarrow \infty$, because, holding $M$ fixed, we will always be able to find $M$ observations in the control group that are arbitrarily close to $\bar x$. Independence of  $\kappa_i$ follows from the fact that the probability of two treated observations sharing the same nearest neighbor converges to zero. See details in  Appendix \ref{proof_unbiased}.

Proposition \ref{unbiased} shows that the expected value of the NN matching estimator converges to $\tau$.   We also derive in Proposition \ref{unbiased}  the asymptotic distribution of $\hat \tau$, which has expected value equal to $\tau$. Therefore, the NN matching estimator is asymptotically unbiased whether we define asymptotic unbiasedness as  $\mathbb{E}[\hat \tau] \rightarrow \tau$, or as  $\hat \tau  \buildrel d \over \rightarrow \tau + \widetilde \kappa$, with $\mathbb{E}[\widetilde \kappa]=0$.

\begin{remark}
\normalfont
 With $N_1$ fixed, the estimator is not consistent. This happens because, with $N_1$ fixed, we cannot apply a law of large numbers to the  average of the error of the treated observations. For the same reason, the matching estimator will not generally be asymptotically normal.  These conclusions are similar to the ones derived by \cite{CT} for  differences-in-differences estimators with few treated groups.

\end{remark}

\begin{remark} \label{remark_bias}
\normalfont
Consider a bias-corrected estimator suggested by \cite{AI_2011},
\begin{eqnarray} \label{biasadj}
\hat \tau_{biasadj} =  \frac{1}{N_1} \sum_{i \in \mathcal{I}_1} \left[ Y_i - \frac{1}{M} \sum_{m=1}^M \left( Y_j + \hat \mu_0(X_i) - \hat \mu_0(X^i_{(m)})  \right) \right],
\end{eqnarray}
where $\hat \mu_0(x)$ is an estimator for $\mu_0(x) $, and $X^i_{(m)}$ be the covariate value of the $m$-closest match to observation $i$. If the conditions on $\hat \mu_0(x)$ considered by \cite{AI_2011} are satisfied, then we can also guarantee that  $\hat \tau_{biasadj} $ has the same asymptotic distribution as $\hat \tau$. The intuition is that $\hat \mu_0(X_i) - \hat \mu_0(X^i_{(m)})$ converges in probability to zero when $N_0 \rightarrow \infty$, because $X^i_{(m)} \buildrel p \over \rightarrow X_i$. 

\end{remark}

\begin{remark} 
\normalfont

We consider an asymptotic framework in which $M$ is held fixed, while $N_0 \rightarrow \infty$, which is similar to what AI call fixed-$M$ asymptotics in their setting. As argued by AI, the motivation for such fixed-$M$ asymptotics is to provide an approximation to the sampling distribution of matching estimators with a small number of matches. Matching estimators using few matches have been widely used in applied work (see AI). Moreover, \cite{Imbens_Rubin} argue  against using matching estimators with many matches, as this would tend to increase the bias of the resulting estimator, while the marginal gains in precision of increasing the number of matches are limited. In Section \ref{conclusion} we consider the implication of our findings for other types of matching estimators.

\end{remark}

\section{ Inference  } \label{inference}

The fact that the NN matching estimator is not generally asymptotically normal when $N_1$ is fixed and $N_0 \rightarrow \infty$ poses an important challenge when it comes to inference. In particular, inference based on the asymptotically normal distribution  derived by AI, or on the bootstrap procedure suggested by \cite{Otsu}, should not provide a good approximation in our setting, as the asymptotic theory behind these methods relies on both $N_1$ and $N_0$ going to infinity. We therefore consider  alternative inference methods  based on the theory of randomization tests under an approximate symmetry assumption, developed by \cite{Canay}. We focus on a test based on sign changes, while in in Appendix  \ref{Section_permutation} we consider an alternative test based on permutations.  We consider the problem of testing the null hypothesis $H_0: \tau = c$.

Without loss of generality, let $i=1,...,N_1$ be the treated observations, and consider a function of the data given by
\begin{eqnarray}
S_{N_0} = \left( \hat \tau_1^{N_0},...,\hat \tau_{N_1}^{N_0}   \right)'
\end{eqnarray}
where $\hat \tau_i^{N_0} = \left(Y_i - \frac{1}{M}  \sum_{j \in \mathcal{J}_M(i)} Y_j \right)- c$. Each $\hat \tau_i^{N_0}$ depends on the $M$ nearest neighbors of observation $i$, so its distribution  depends on $N_0$.

Following \cite{Canay}, we consider a test statistic given by
\begin{eqnarray}
T(S_{N_0}) = \frac{| \hat \tau^{N_0} | }{\sqrt{\frac{1}{N_1-1}\sum_{i=1}^{N_1}(\hat \tau_i^{N_0} - \hat \tau)^2}},
\end{eqnarray}
where  $\hat \tau^{N_0} = \frac{1}{N_1} \sum_{i \in  \mathcal{I}_1} \hat \tau_i^{N_0}$. Note that $\hat \tau^{N_0} = \hat \tau$ if $c=0$.

We consider  the group of transformations given by $\textbf{G} = \{ -1,1 \}^{N_1}$, where $gS_{N_0} = \left( g_1\hat \tau_1^{N_0},...,g_{N_1}\hat \tau_{N_1}^{N_0}   \right)'$. Let $K = |\textbf{G}|$ and denote by
\begin{eqnarray}
T^{(1)}(S_{N_0}) \leq T^{(2)}(S_{N_0}) \leq ... \leq T^{(K)}(S_{N_0})
\end{eqnarray}
the ordered values of $\{ T(gS_{N_0}) : g \in \textbf{G} \}$. Let $k=\lceil K(1-\alpha) \rceil$, where $\alpha$ is the significance level of the test. Then the  test is given by
\begin{eqnarray}  \label{decision}
\phi(S_{N_0}) = \begin{cases} 1 & \mbox{ if } T(S_{N_0}) > T^{(k)}(S_{N_0})    \\ 0 & \mbox{ if } T(S_{N_0}) \leq T^{(k)}(S_{N_0}).   \end{cases}
\end{eqnarray}

In words, we calculate the test statistic $T(gS_{N_0} )$ for all possible $gS_{N_0} =  \left( g_1\hat \tau_1^{N_0},...,g_{N_1}\hat \tau_{N_1}^{N_0}   \right)'$, and then we compare the actual test statistic  $T(S_{N_0} )$ with the distribution $\{  T(gS_{N_0} ) : g \in \textbf{G}  \}$. We first show validity of such test when $N_0 \rightarrow \infty$ and $N_1$ is fixed under symmetry conditions on the distribution of potential outcomes and on the distribution of heterogeneous treatment effects.

\begin{assumption}[Symmetry] \label{Assumption_symmetry}
\label{Assumption_symmetry} (i) $Y(w) | (X=x,W=w)$ is symmetric around its mean for $w \in \{0,1\}$ and for all $x \in \mathbb{X}_1$, and (ii) the distribution of $\mu_1(X) - \mu_0(X)$ conditional on $W=1$ is symmetric around $ \tau$.

\end{assumption}

While this is a strong assumption, the condition that  potential outcomes, conditional on $X$, are symmetric can be justified in settings in which observation $i$ is the average of a large number of individuals, by appealing to some central limit theorem.\footnote{Notice that  this does not preclude dependence between individuals in observation $i$, insofar as it is still amenable to a central limit theorem.} This could be the case, for example, in our empirical application in which each observation represents average test scores of a large number of students per school, even if we do not observe student-level data. While we cannot test the plausibility of this assumption for $Y(1)$ (given that we have fixed $N_1$) and for the distribution of heterogeneous treatment effects, we can provide evidence on whether the distribution of $Y(0) | (X=x,W=0)$ is symmetric by fitting a model for $\mu_0(x)$ and checking whether the residuals are symmetric for the controls.

We show that the sign-changes test is asymptotically valid when $N_0 \rightarrow \infty$  under such symmetry assumptions, even when $N_1$ is fixed.  Our MC simulations presented in Section \ref{MC_simulations} suggest that relaxing Assumption \ref{Assumption_symmetry} does not generate large size distortions for this test, except in settings in which $N_1$ is very small, and the asymmetry in the potential outcomes or heterogeneous effects is very strong.

\begin{proposition} \label{test}

Suppose Assumptions \ref{sample}, \ref{CIA},  \ref{overlap},  \ref{assumption_mu}(b), and  \ref{Assumption_symmetry} hold. Assume also that the distribution of $Y$ is continuous. If we consider the problem of testing $H_0:  \tau =c$, then, for any $\alpha \in (0,1)$, $\mbox{limsup}_{N_0 \rightarrow \infty} \mathbb{E} \left[  \phi( S_{N_0}) \right]  \leq \alpha$ when $N_1$ is fixed. 

\end{proposition}

The main idea of the proof is to show that  the limiting distribution of $S_{N_0}$, under the null, is invariant  to the transformations in ${\textbf{G}}$. This is true because, asymptotically, $\hat \tau_i^{N_0}$ and $\hat \tau_j^{N_0}$ are independent for $i \neq j$, and, under the null, $\hat \tau_i^{N_0}$ converges in distribution to $\kappa_i$, which is  symmetric around zero given Assumption \ref{Assumption_symmetry}. Details  in  Appendix \ref{proof_sign_changes_asympt}.

While Proposition \ref{test} provides a test that is valid when $N_1$ is fixed, validity even when $N_1$ is fixed comes at a cost of relying on stronger assumptions than usually considered in the matching literature. We show that we can relax Assumption \ref{Assumption_symmetry} if  $N_1$ increases, but at a slower rate relative to $N_0$.  We consider the following assumptions, which are similar to the ones considered by AI for the setting in which $N_0$ grows at a faster rate than $N_1$.

\begin{assumption}[Sampling rates]
\label{Assumption_rates}
For some $r > \mbox{max}\{k/2,2\}, N_{1}^{r} / N_{0} \rightarrow \theta$ with $0<\theta<\infty$.

\end{assumption}

\begin{assumption}[Distribution of potential outcomes]
\label{Assumption_appendix2}

 For $w=0,1,$ (i) $\mu(x, w)$ and $\sigma^{2}(x, w)$ are Lipschitz in $\mathbb{X}_0$, (ii) for some $\gamma>0$, $\mathbb{E}\left[ |\epsilon|^{4+\gamma} | W=w,X = x \right]$  exists and is bounded uniformly in $x,$ and (iii) $\sigma^{2}(x, w)$ is bounded away from zero.

\end{assumption}

Under these conditions, Corollary 1(ii) from AI implies that the NN matching estimator is consistent and asymptotically normal. We show that the sign-changes test is asymptotically valid in this setting in which $N_1$ increases (but at a lower rate than $N_0$) even when we relax Assumption \ref{Assumption_symmetry}.

\begin{proposition} \label{sign_changes_asympt}

Suppose Assumptions \ref{sample}, \ref{CIA}, \ref{overlap},  \ref{Assumption_rates}, and \ref{Assumption_appendix2} hold.  If we consider the problem of testing $H_0: \tau =c$, then the sign-changes test is asymptotically valid when  $N_1,N_0 \rightarrow \infty$.

\end{proposition}

Therefore, the sign-changes test is asymptotically valid under weaker conditions when $N_1$ also increases. The condition that $N_0$ grows at a faster rate than $N_1$ (Assumption \ref{Assumption_rates}) is important for two reasons. First, it guarantees that we can apply Corollary 1(ii) from AI, which implies that the bias of the NN matching estimator is asymptotically negligible. Second, it also guarantees that the probability that we have shared nearest neighbors converges to zero when $N_1,N_0 \rightarrow \infty$. See details of the proof in  Appendix \ref{proof_sign_changes_asympt}.

\begin{remark}
\normalfont 

We can construct confidence intervals considering the set of values $c$ such that the null would not be rejected. Propositions \ref{test} and \ref{sign_changes_asympt} provide the assumptions we need for validity of such confidence intervals whether we consider a setting with $N_1$ fixed or $N_1 \rightarrow \infty$.

\end{remark}

\begin{remark}
\normalfont \label{finite_sample_adjustment2}
In  Propositions \ref{test} and  \ref{sign_changes_asympt}, this test is asymptotically valid because the probability that different treated observations share the same nearest neighbor goes to zero, when $N_0 \rightarrow \infty$. If there are shared nearest neighbors in finite samples, this may lead to over-rejection if we do not take that into account.  Therefore, we  suggest a finite sample adjustment, in which  we restrict to sign changes such that $g_i=g_j$ if $i$ and $j$ share the same nearest neighbor. The probability that this modification is relevant converges to zero when $N_0 \rightarrow \infty$.\footnote{Another alternative would be to consider a matching estimator without replacement. However, this would generate lower quality matches, which implies  more bias (AI). Moreover, matching without replacement has the disadvantage that the estimator is not invariant to different sorting of the data.   } 

\end{remark}

\begin{remark} \label{ties}
\normalfont
\cite{Canay} consider a randomized version of the test to deal with cases such that  $ T( S_{N_1}) = T^{( k)}( S_{N_1})$, while we consider a test that rejects  if $ T ( S_{N_1}) >  T^{( k)}( S_{N_1})$. Such randomization guarantees an  asymptotic size of $\alpha$ even when $N_1$ is fixed. 

\end{remark}

\begin{remark}
\normalfont
This test is also asymptotically valid for bias-corrected matching estimators, as defined in equation (\ref{biasadj}). We define  $\tilde \tau_{i}^{N_0} = Y_i - \frac{1}{M}  \sum_{j \in \mathcal{J}_M(i)} \left( Y_j - \hat \mu_0(X_j) + \hat \mu_0(X_i) \right)$ in this case. We also consider the implications of Proposition \ref{test} for other types of matching estimators in Section \ref{conclusion}.

\end{remark}

\section{Monte Carlo Simulations} \label{MC}

We present two sets of MC simulations. First, we present an empirical MC simulation, which  provides a setting in which there is selection on observables with a structure based on a real application.  Then we consider another set of  MC simulations where the focus is to evaluate the relevance of Assumption \ref{Assumption_symmetry} for the sign-changes test when $N_1$ is small.

\subsection{Empirical Monte Carlo simulations} 
\label{Empirical_MC}

We construct an empirical MC simulation in which treatment assignment and potential outcomes are based on the ``Jovem de Futuro'' program, which we present in more details as an empirical illustration in Section \ref{empirical_application}.\footnote{More details on the ``Jovem de Futuro'' program and on the construction of this empirical MC study are also presented in Appendix \ref{JF}.} 
The main results of the empirical MC simulation are summarized in Table \ref{EMC_text}. Panel A shows that, when we consider NN matching estimators with few nearest neighbors,  the bias of the matching estimator is close to zero, regardless of the number of treated observations. This is true in our simulations even when the number of control observations is not  large.  Increasing the  number of nearest neighbors used in the estimation implies that we need an increasing number of controls to keep  our approximations reliable. We  show in Appendix Table \ref{Table_more_covariates} that increasing the dimensionality of the matching variables also implies that a larger number of controls is needed to keep  our approximations reliable. 

Panels B and C present rejection rates, respectively, for the asymptotic test based on AI and for the sign-changes test.\footnote{In Appendix Table \ref{Table_wild} we consider the  test based on permutations presented in Appendix \ref{Section_permutation} and the wild bootstrap test proposed by \cite{Otsu} as alternative inference methods.} The asymptotic test  generally presents over-rejection when $N_1$ is small, which is consistent with the fact that the theory behind this test relies on $N_1 \rightarrow \infty$. In contrast, the sign-changes test  controls well for size even when $N_1$ is very small. An important caveat, however, is that the sign-changes test may be conservative in settings in which the number of sign-changes transformations is very small, which is a common feature in approximate randomization tests \citep{ART}.  The number of sign-changes transformations will be small when $N_1$ is very small,  or when $N_0$ is small relative to $N_1 \times M$.\footnote{The number of sign-changes  transformations can be small when $N_0$ is small relative to $N_1 \times M$ due to the finite-sample adjustment discussed in Remark \ref{finite_sample_adjustment2}.  } The sign-changes test presents non-trivial power, except for the cases in which it is very conservative (Appendix Table \ref{EMC_power}).

When $N_0$ is large, the sign-changes test presents non-trivial power in these simulations when $N_1=5$ because we consider a 10\%-level test. However, if we considered a 5\%-level test, then it would be very conservative and have a very low power in this case (Appendix Table \ref{Table_MC5}). This happens because we would have very few sign-changes transformations to reject the null at a 5\% significance level. An alternative in  case we want to consider a 5\%-level test with a very small number of treated observations is the approximate randomization test based on permutations, presented in Appendix \ref{Section_permutation}. Similarly to the sign-changes test, the test based on permutations (with the right choice of test statistic) is also valid under stronger assumptions when $N_1$ is fixed, and under weaker assumptions when $N_1$ increases (but at a lower rate than $N_0$). However, it relies on arguably stronger conditions than the sign-changes test when $N_1$ is fixed. We discuss that in more detail  in Section \ref{comparison}.

\subsection{Monte Carlo simulations relaxing symmetry conditions } \label{MC_simulations}

We consider now another set of MC simulations in which we vary de degree of symmetry of the potential outcomes and of the distribution of heterogeneous treatment effects. For all simulations, we set $(X | W=w) \sim N(0,1)$  for $w \in \{ 0,1\}$ and $Y(0) | (X= x,W=0) \sim N(0,1)$ for all $x \in \mathbb{R}$. Therefore, $\mu_0(x)=0$ for all  $x \in \mathbb{R}$.  Then we vary $\mu_1(x)$ and the distribution of $(Y(1) - \mu_1(x))|(X=x,W=1)$. For all settings, we consider $N_1 \in \{5,10,25,50 \}$ and $N_0 = 1000$.

We start considering a setting with $\mu_1(x) = x$ and  $(Y(1) - \mu_1(x))|(X=x,W=1)\sim N(0,1)$, which implies that the symmetry conditions from Assumption \ref{Assumption_symmetry} hold. In this case, we have that $\tau = \mathbb{E}[\mu_1(X) - \mu_0(X) | W=1] = 0$, so the null hypothesis is true. We present rejection rates for 10\% tests in Panel A of Table \ref{Table_MC2}. Consistent with the MC simulations from Section \ref{Empirical_MC}, the asymptotic test based on AI over rejects when $N_1$ is small. In contrast, the sign-changes test does not over-reject irrespectively of $N_1$. This is expected, because  Assumption \ref{Assumption_symmetry} is valid in this case, and we consider a setting in which the estimator is unbiased.

 In Figure \ref{Figure_MC}, we contrast the power of these two tests when $M=4$, for different values of $N_1$. We modify the DGP so that $\mu_1(x) = \tau + x$, implying that $\mathbb{E}[\mu_1(X) - \mu_0(X)|W=1] = \tau$.  We present size-adjusted power for the AI test  using critical values that set rejection rates equal to 10\% in the MC simulations when the null is true. Since the sign-changes test never over-rejects in this setting, rejection rates are not adjusted for this test.  When $N_1=5$, the sign-changes test presents non-trivial power, but its  power is lower than the size-adjusted power of the asymptotic test. We recall, however, that the  asymptotic test presents large over-rejection in this scenario, so it is not a feasible alternative. When $N_1 = 10$, the loss in power of the sign-changes test is very small, while the asymptotic test still presents relevant over-rejection. 

When we further increase $N_1$, the size distortion of the asymptotic test diminishes. In this case, the sign-changes test continues to control for size, but presents a small loss in power relative to the asymptotic test when $N_1$ increases.  This happens because, in this case, $N_0$ becomes smaller relative to $N_1 \times M$. Consider the case in which $N_1=50$, $N_0=1000$, and $\tau = 0.5$. In this case, the size-adjusted power of the asymptotic test is $0.756$, while the power of the sign-changes test is $0.690$. If we increase $N_0$ to $2000$, then the gap in power   between these tests goes down from 5.6pp to 3.7pp. In contrast, if we reduce $N_0$ to 500, then the gap in power increases to 11pp. Therefore, the sign-changes test does not present relevant losses in terms of power if $N_0$ is large, but may present some loss in power if $N_0$ is not much larger than $N_1 \times M$.

In panels B to E in Table   \ref{Table_MC2}, we consider variations in the DGP in which the symmetry conditions from Assumption \ref{Assumption_symmetry} does not hold. We first set $\mu_1(x) = Q^{-1}(\Phi(x);p)-p)/\sqrt{2p}$, where  $Q(x;p)$ is the CDF of a chi-squared distribution with $p \in \mathbb{N}$ degrees of freedom, and $\Phi(x)$ is the CDF of a standard normal. In this case, we have that $(\mu_1(X) - \mu_0(X)) | W=1$ has mean zero and variance one, but its distribution is asymmetric. The asymmetry is decreasing with $p$. We also consider settings in which   $(Y(1) - \mu_1(x))|(X=x,W=1) \sim (\chi^2_p - p)/\sqrt{2p}$ instead of standard normal, which adds more asymmetry in the distribution of $\kappa_i$. Again, this distribution has mean zero and variance one, but it has an asymmetry that is decreasing in $p$.  In Appendix Table \ref{Table_MC1}, we also consider  cases in which $\mu_1(x)=0$ and we only vary the distribution of  $(Y(1) - \mu_1(x))|(X=x,W=1)$. In these simulations, the sign-changes test continues to control for size, except when $N_1$ is small and the distribution of $\kappa_i$ is very asymmetric. Importantly, even in the scenarios in which the sign-changes test presents some over-rejection, its over-rejection is milder relative to the over-rejection of the asymptotic test. When $N_1$ increases, then both tests control for size, which is consistent with the fact that they are asymptotically valid when $N_1 \rightarrow \infty$, even when $\kappa_i$ is asymmetric.  

Overall, based on these simulations, the sign-changes test presents important gains in terms of test size when $N_1$ is small, even when the distribution of $\kappa_i$ is asymmetric. Except when this distribution is extremely asymmetric, the sign-changes test does not present much over-rejection. Moreover, even when it presents some over-rejection in these simulations, the over-rejection is milder relative to the over-rejection of the asymptotic test. Finally, in settings in which the asymptotic test controls well  for size, the cost in terms of power for the sign-changes test  is low, as long as $N_0$ is sufficiently large relative to $N_1 \times M$. Therefore, this test provides an interesting alternative for settings in which $N_1$ is small, and also when $N_1$ is not very small, but  $N_0 >> N_1$.

\section{Comparing Alternative Inference Methods} \label{comparison}

The different test procedures we consider potentially present important trade-offs in terms of size distortion and power, depending on the number of treated and control observations. Moreover, the sign-changes test  relies on different sets of assumptions depending on whether the empirical application is better approximated by a theory in which $N_1$ is fixed or in which $N_1$ diverges (but at a slower rate relative to $N_0$). In light of the theoretical properties derived in Section \ref{inference}, and of the evidence from the MC simulations presented in Section \ref{MC}, we provide guidance on how to consider the suitability of different inference methods in empirical applications. 

If $N_1$ is large, then the asymptotic approximations considered by  AI should be reliable. In this case, if we are in a setting in which $N_0$ is much larger than $N_1$, then the sign-changes test would be comparable to the asymptotic test. More specifically, both tests would be valid under the same assumptions regarding the distributions of potential outcomes, and they would have similar power. However, if $N_0$ is not very large relative to $M \times N_1$, then the sign-changes test  may have lower power, and the asymptotic test should be preferable.  

If $N_1$ is not very large, then the test based on AI presents relevant  size distortions, and the sign-changes test  becomes an interesting alternative. In this case, one should be aware that this test is valid under stronger assumptions if $N_1$ is very small, so these assumptions should be discussed by applied researchers.  Since this test is asymptotically valid even when we relax such symmetry conditions when $N_1$ increases, we expect that distortions in case such assumptions are not valid to be relatively minor, except in cases in which $N_1$ is very small and errors or treatment effects are very asymmetric. This intuition is corroborated by the simulations presented in Section \ref{MC_simulations}. In those simulations,  the sign-changes test only presents relevant size distortions when $N_1$ is very small and the degree of asymmetry in the distribution of $\kappa_i$ is large. Moreover, in such settings, the asymptotic test based on AI presents more severe size distortions than the sign-changes test. 

Overall, if $N_1$ is not very large, then the sign-changes test presents relevant gains relative to the asymptotic test in terms of controlling for test size. Moreover, if  $N_0$ is large relative to $N_1 \times M$, then the sign-changes test has a power comparable to the (size-adjusted) power of the  asymptotic test. {The only exception in which the sign-changes test would have a lower power than the asymptotic test even when  $N_0$ is large is when $N_1$ is very small (for example, when  $N_1=5$). However, those are exactly the cases in which we should expect the size distortions of the asymptotic test to be more severe. } 

Finally, it is worth noting that the sign-changes test only presents non-trivial power in settings in which  $N_1$ is very small (say, $N_1 = 5$), when we consider 10\%-level tests. A feasible alternative when $N_1$ is even smaller than 5, or when we want to consider a 5\%-level test, is the test based on permutations described in Appendix \ref{Section_permutation}. However, one should be aware that, with $N_1$ fixed, such test would rely on homoskedasticity and treatment effects homogeneity assumptions, which are arguably stronger than Assumption \ref{Assumption_symmetry}. Similarly to the sign-changes test, the test based on permutations, with the right choice of test statistic, is also valid under  weaker assumptions when $N_1$ increases (but at a lower rate than $N_0$).

\section{Empirical Illustration} \label{empirical_application}

As an empirical illustration of the NN matching estimator in a setting with small $N_1$ relative to $N_0$, we analyze the ``Jovem de Futuro'' program.  This is a program that has been running in Brazil since 2008, aimed at improving the quality of education in public schools by improving management practices and allocating grants to  treated schools. In 2010, this program was implemented in a randomized control trial with 15 treated  schools in Rio de Janeiro and 39 treated schools in Sao Paulo, with the same number of control schools in each state. In Appendix \ref{JF} we present more details on this empirical application. We rely on this randomized control trial to validate the use of NN matching estimators in a setting in which there are few treated and many control units.\footnote{Influential  papers that evaluate the use of non-experimental methods in empirical applications where a randomized control trial is available include \cite{10.2307/2677743}, \cite{Heckman13416}, \cite{10.2307/2999630}, \cite{10.2307/2971733}, \cite{ASMITH2005305}, \cite{Lalonde}, \cite{DW1999}, and \cite{DW2002}.} We take advantage of the fact that there were about 1,000 other public schools in Rio de Janeiro and more than 3,000 other public schools in Sao Paulo that did not participate in the experiment. More specifically, we consider a NN matching estimator using the experimental \textit{control} schools as treated observations, and schools that did not participate in the experiment as control observations. These experimental control schools were selected following the same process used for the selection of treated schools. However, since these schools did not actually receive the treatment in the analyzed period, we should not expect to find significant effects in this case if the matching estimator is valid  \citep{10.2307/2677743}.  Therefore, this provides an interesting setting to evaluate the validity of matching estimators with few treated and many control observations.

Table \ref{empirical_aplication_control_text} shows estimated effects from 2010 to 2012. We use test scores from 2007 to 2009 as matching variables.   In addition to the point estimates,  p-values are calculated using the asymptotic distribution derived by AI, and from the sign-changes test.  Interestingly, estimates for Rio de Janeiro (columns 1 to 4)  generally have lower p-values using the test based on AI, relative to the alternative inference procedure. In particular, a test based on the asymptotic distribution  would reject the null at 10\% in two cases, while the sign-changes  test would fail to reject the null. This is consistent with our  simulations from Section \ref{MC}, where we  show that  the asymptotic test based on AI may lead to over-rejection when $N_1$ is small. The difference in p-values across different methods is less pronounced when we consider estimates for Sao Paulo, which is consistent with having a larger number of ``treated'' schools in Sao Paulo.

\section{Conclusion} \label{conclusion}

We consider the asymptotic properties of matching estimators when the number of control observations is large, but the number of treated observations is fixed. In this setting, the NN matching estimator is asymptotically unbiased for the ATT under standard assumptions used in the literature on estimation of treatment effects under selection on unobservables. Moreover, we provide tests, based on the theory of randomization under approximate symmetry, that are asymptotically valid when the number of treated observations is fixed and the number of control observations goes to infinity. While we need to rely on relatively strong assumptions so that these tests are valid even when $N_1$ is fixed, we show that these tests are also valid under weaker assumptions when $N_1$ increases, but at a slower rate relative to $N_0$.  We analyze in details the advantages and disadvantages of these inference methods, and provide guidance on which methods should be used in specific applications.

We conjecture that the asymptotic unbiasedness and the asymptotic validity of the randomization inference test based on sign changes when $N_1$ is fixed remain valid if we consider other types of matching estimators. Intuitively, the main requirement should be that, when constructing the counter-factual for a treated unit $j \in \mathcal{I}_1$, the estimator would rely on a weighted average of the control observations such that, as $N_0 \rightarrow \infty$, an increasing proportion of the weights would be allocated to control units $i$ with $X_i$ close to $X_j$. This would be true if we use, for example, a kernel method for the weights under suitable conditions on the smoothing parameter \citep{10.2307/2971733}. In this case, the bias of the treatment effect estimator for each treated observation $j \in \mathcal{I}_1$ would go to zero as $N_0 \rightarrow \infty$. Moreover, the correlation between the treatment effect estimator for different treated observations would also go to zero in this case, providing asymptotic validity for the inference method based on sign changes. However, an advantage of considering the NN matching estimator in this setting is that it would be more straightforward to implement the adjustment proposed in Remark \ref{finite_sample_adjustment2} to avoid over-rejection in finite samples for the inference method based on sign changes.  Moreover, it would also be possible to consider the randomization test based on permutations, presented in Appendix \ref{Section_permutation}, when we rely on NN matching estimators.

Our results are also relevant for synthetic control (SC) applications. Following \cite{Doudchenko}, the SC and the matching estimators are nested in a framework in which the estimated counterfactual outcome for the treated observation is a linear combination of the outcomes for the controls. In their framework,  consider an estimator in which the weights given to control observations with large discrepancies in pre-treatment outcomes relative to the treated units go to zero. In this case, following the same arguments as above, the estimator would be asymptotically unbiased  if treatment assignment is ``as good as random,'' conditional on this set of pre-treatment outcomes.\footnote{See \cite{Abadie2010}, \cite{FB}, \cite{FP_SC},  and \cite{ferman} for a discussion on the validity of the synthetic control estimator under a different set of assumptions.} This is exactly the case for the penalized SC estimator for disaggregated data proposed by \cite{lhour}. Under these conditions, the randomization inference test we propose based on sign changes  remains asymptotically valid when the number of control units goes to infinity. This provides an interesting alternative for inference, when there are multiple treated units and a large number of control units, that does not rely on exchangeability nor homoskedasticity assumptions.\footnote{See \cite{Firpo2018SyntheticCM}, \cite{FP_tests}  and \cite{Hahn} for a discussion on the  placebo test proposed by \cite{Abadie2010}. \cite{Victor} propose a permutation test based on the timing of the intervention. This test, however, would require a very large number of periods.  Instead, our test may be an alternative when the number of periods is not large, but the number of control units is large. }    The only caveat is that  a very large number of control observations is needed when the number of pre-treatment periods is large,  so that  approximations remain reliable.



\singlespace

\renewcommand{\refname}{References} 

\bibliographystyle{apalike}
\bibliography{bib/bib.bib}

\pagebreak

\begin{table}[H]
 \begin{center}
\caption{{\bf Empirical MC simulation}} \label{EMC_text}
\begin{tabular}{lcccccccc}
\hline
\hline

     & \multicolumn{2}{c}{$M=1$}   & & \multicolumn{2}{c}{$M=4$} & & \multicolumn{2}{c}{$M=10$ }     \\ \cline{2-3} \cline{5-6}  \cline{8-9} 
     & $N_0=50$ & $N_0=500$ & & $N_0=50$ & $N_0=500$   &  & $N_0=50$ & $N_0=500$ \\

     & (1) &  (2)  & & (3) &  (4) & & (5) & (6)  \\ 
     
     \hline

\multicolumn{9}{c}{\textit{Panel A: $| \mbox{average   bias} \times 100 | $}} \\

$N_1=5$ & 1.143 & 0.338 &  & 1.618 & 0.673 &  & 2.156 & 0.936 \\
$N_1=10$ & 1.112 & 0.465 &  & 1.585 & 0.711 &  & 2.085 & 0.706 \\
$N_1=25$ & 0.883 & 0.369 &  & 1.547 & 0.576 &  & 2.148 & 0.833 \\
$N_1=50$ & 1.030 & 0.466 &  & 1.608 & 0.635 &  & 2.137 & 0.771 \\

\\
\multicolumn{9}{c}{\textit{Panel B: rejection rates based on AI   }} \\

$N_1=5$ & 0.204 & 0.210 &  & 0.206 & 0.209 &  & 0.203 & 0.206 \\
$N_1=10$ & 0.151 & 0.160 &  & 0.146 & 0.148 &  & 0.156 & 0.151 \\
$N_1=25$ & 0.123 & 0.121 &  & 0.120 & 0.124 &  & 0.135 & 0.127 \\
$N_1=50$ & 0.120 & 0.107 &  & 0.125 & 0.117 &  & 0.144 & 0.117 \\
\\

\multicolumn{9}{c}{\textit{Panel C: test based on RI, sign changes }} \\
$N_1=5$ & 0.048 & 0.067 &  & 0.004 & 0.053 &  & 0.000 & 0.030 \\
$N_1=10$ & 0.098 & 0.105 &  & 0.005 & 0.095 &  & 0.000 & 0.064 \\
$N_1=25$ & 0.104 & 0.099 &  & 0.000 & 0.099 &  & 0.000 & 0.065 \\
$N_1=50$ & 0.106 & 0.096 &  & 0.000 & 0.104 &  & 0.000 & 0.021 \\

\hline
\end{tabular}
\end{center}
\footnotesize Note: This table presents simulation results from the empirical MC study described in details in Supplemental Appendix D.  Panel A reports the average bias  (multiplied by 100). The bias if we considered a naive comparison between treated and control schools would be, in expectation, $-0.32$. Panels B  and C present rejection rates for 10\%-level tests. Panel  B is based on the asymptotic distribution derived by AI, while Panel C presents rejection rates for the randomization inference test based on sign changes.  For each combination $(N_1,N_0)$, we run 10,000 simulations.  

\end{table}



\begin{table}[H]
 \begin{center}
\caption{{\bf MC simulation: relaxing symmetry conditions}} \label{Table_MC2}
\begin{tabular}{lcccccccc}
\hline
\hline

     & \multicolumn{2}{c}{$M=1$}   & & \multicolumn{2}{c}{$M=4$} & & \multicolumn{2}{c}{$M=10$ }     \\ \cline{2-3} \cline{5-6}  \cline{8-9} 
     & AI & Sign-changes &   & AI & Sign-changes  &   & AI & Sign-changes \\

     & (1) &  (2)  & & (3) &  (4) & & (5) & (6)  \\ 
     \hline
     
\multicolumn{9}{c}{Panel A: $\mu_1(x) = x$ and $\epsilon|(X=x,W=1) \sim N(0,1)$ ($\kappa_i$ symmetric)} \\
$N_1=5$ & 0.213 & 0.067 &  & 0.216 & 0.065 &  & 0.214 & 0.057 \\
$N_1=10$ & 0.151 & 0.098 &  & 0.152 & 0.092 &  & 0.155 & 0.083 \\
$N_1=25$ & 0.122 & 0.100 &  & 0.119 & 0.088 &  & 0.121 & 0.068 \\
$N_1=50$ & 0.104 & 0.091 &  & 0.110 & 0.073 &  & 0.109 & 0.035 \\
\\
\multicolumn{9}{c}{Panel B: $\mu_1(x) = \frac{1}{4}(Q^{-1}(\Phi(x);8)-8)$   and  $\epsilon|(X=x,W=1) \sim N(0,1)$} \\
$N_1=5$ & 0.213 & 0.067 &  & 0.219 & 0.067 &  & 0.219 & 0.055 \\
$N_1=10$ & 0.152 & 0.096 &  & 0.153 & 0.094 &  & 0.155 & 0.084 \\
$N_1=25$ & 0.124 & 0.103 &  & 0.123 & 0.090 &  & 0.124 & 0.068 \\
$N_1=50$ & 0.105 & 0.092 &  & 0.111 & 0.076 &  & 0.109 & 0.042 \\
\\
\multicolumn{9}{c}{Panel C: $\mu_1(x) = \frac{1}{\sqrt{2}}(Q^{-1}(\Phi(x);1)-1)$   and  $\epsilon|(X=x,W=1) \sim N(0,1)$} \\
$N_1=5$ & 0.212 & 0.070 &  & 0.217 & 0.066 &  & 0.216 & 0.059 \\
$N_1=10$ & 0.155 & 0.107 &  & 0.155 & 0.097 &  & 0.157 & 0.095 \\
$N_1=25$ & 0.126 & 0.105 &  & 0.126 & 0.097 &  & 0.131 & 0.080 \\
$N_1=50$ & 0.108 & 0.093 &  & 0.114 & 0.082 &  & 0.113 & 0.054 \\
\\
\multicolumn{9}{c}{Panel D: $\mu_1(x) = \frac{1}{\sqrt{2}}(Q^{-1}(\Phi(x);1)-1)$   and  $\epsilon|(X=x,W=1) \sim \frac{1}{\sqrt{2}}(\chi^2_1 - 1)$} \\
$N_1=5$ & 0.229 & 0.078 &  & 0.253 & 0.087 &  & 0.262 & 0.082 \\
$N_1=10$ & 0.172 & 0.113 &  & 0.179 & 0.124 &  & 0.191 & 0.120 \\
$N_1=25$ & 0.130 & 0.107 &  & 0.134 & 0.108 &  & 0.138 & 0.094 \\
$N_1=50$ & 0.112 & 0.099 &  & 0.120 & 0.082 &  & 0.119 & 0.054 \\
\\
\multicolumn{9}{c}{Panel E: $\mu_1(x) = \frac{1}{\sqrt{2}}(Q^{-1}(\Phi(x);1)-1)$   and  $\epsilon|(X=x,W=1) \sim \frac{2}{\sqrt{2}}(\chi^2_1 - 1)$} \\
$N_1=5$ & 0.248 & 0.092 &  & 0.267 & 0.107 &  & 0.272 & 0.100 \\
$N_1=10$ & 0.185 & 0.128 &  & 0.194 & 0.137 &  & 0.199 & 0.136 \\
$N_1=25$ & 0.140 & 0.118 &  & 0.152 & 0.123 &  & 0.150 & 0.115 \\
$N_1=50$ & 0.120 & 0.110 &  & 0.123 & 0.100 &  & 0.125 & 0.079 \\

\hline
\end{tabular}
\end{center}
\footnotesize Note: This table presents rejection rates for 10\%-level tests for the MC simulations discussed in Section \ref{MC_simulations}. We present rejection rates based on the asymptotic test derived by AI and based on the sign-changes test. In all simulations, $X |(W=w) \sim N(0,1)$ for $w \in \{0,1\}$, $Y(0)|(X=x,W=0) \sim N(0,1)$ for all $x \in \mathbb{R}$, and $N_0=1000$. Each panel presents results for different functions $\mu_1(x)$ and different distributions of $(Y(1)-\mu_1(x)|(X=x,W=1)$. The implied distribution of $\kappa_i$ is symmetric in Panel A, and becomes more asymmetric when we go to Panel E. For each cell, we run 5000 simulations.   

\end{table}


\begin{figure}[H] 

\begin{center}
\caption{{\bf MC Simulations with heterogeneous treatment effects - Power}}  \label{Figure_MC}

\begin{tabular}{cc}

Figure \ref{Figure_MC}.A: $N_1 = 5$ & Figure \ref{Figure_MC}.B:  $N_1 = 10 $ \\

\includegraphics[scale=0.5]{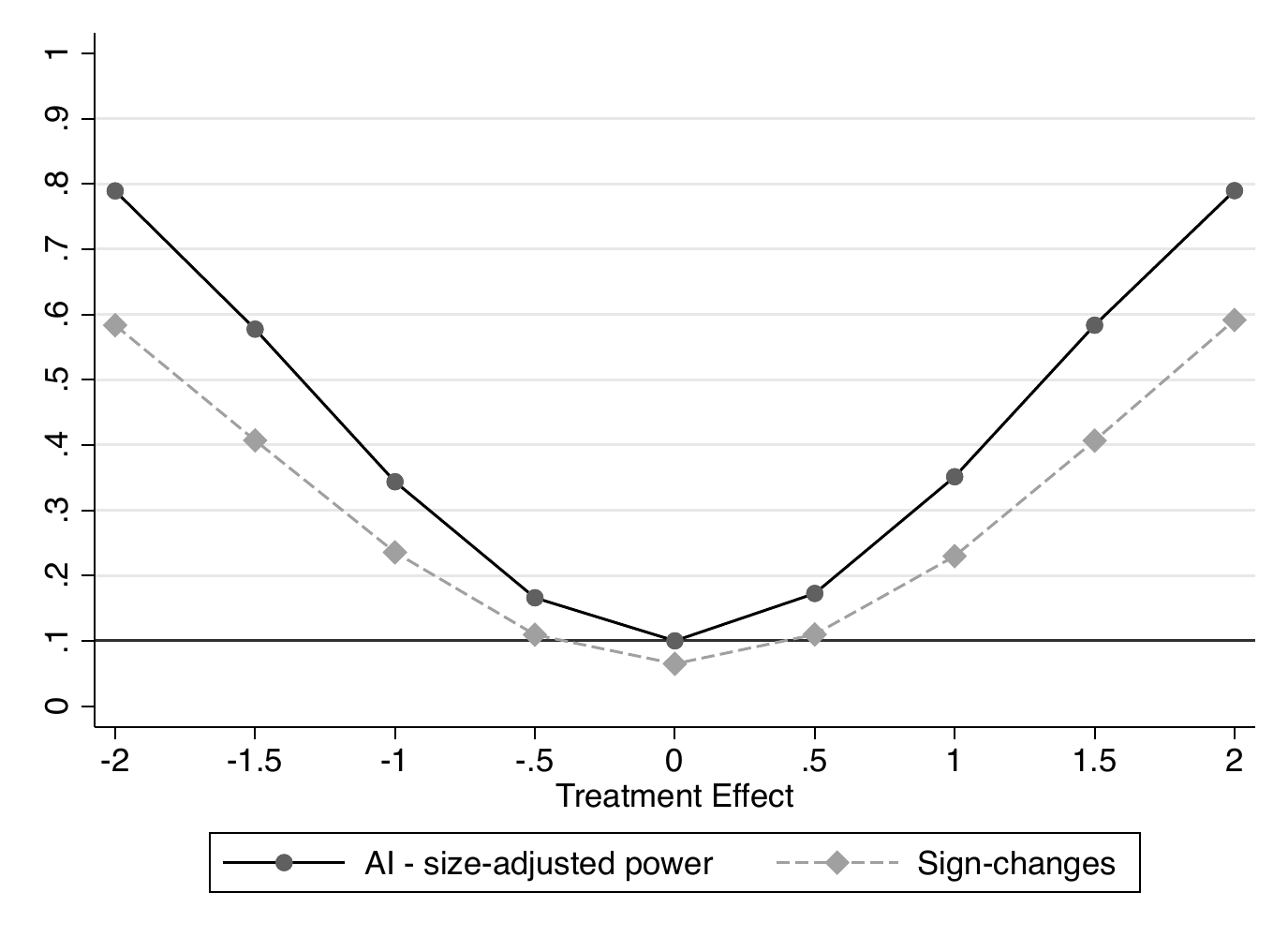} & \includegraphics[scale=0.5]{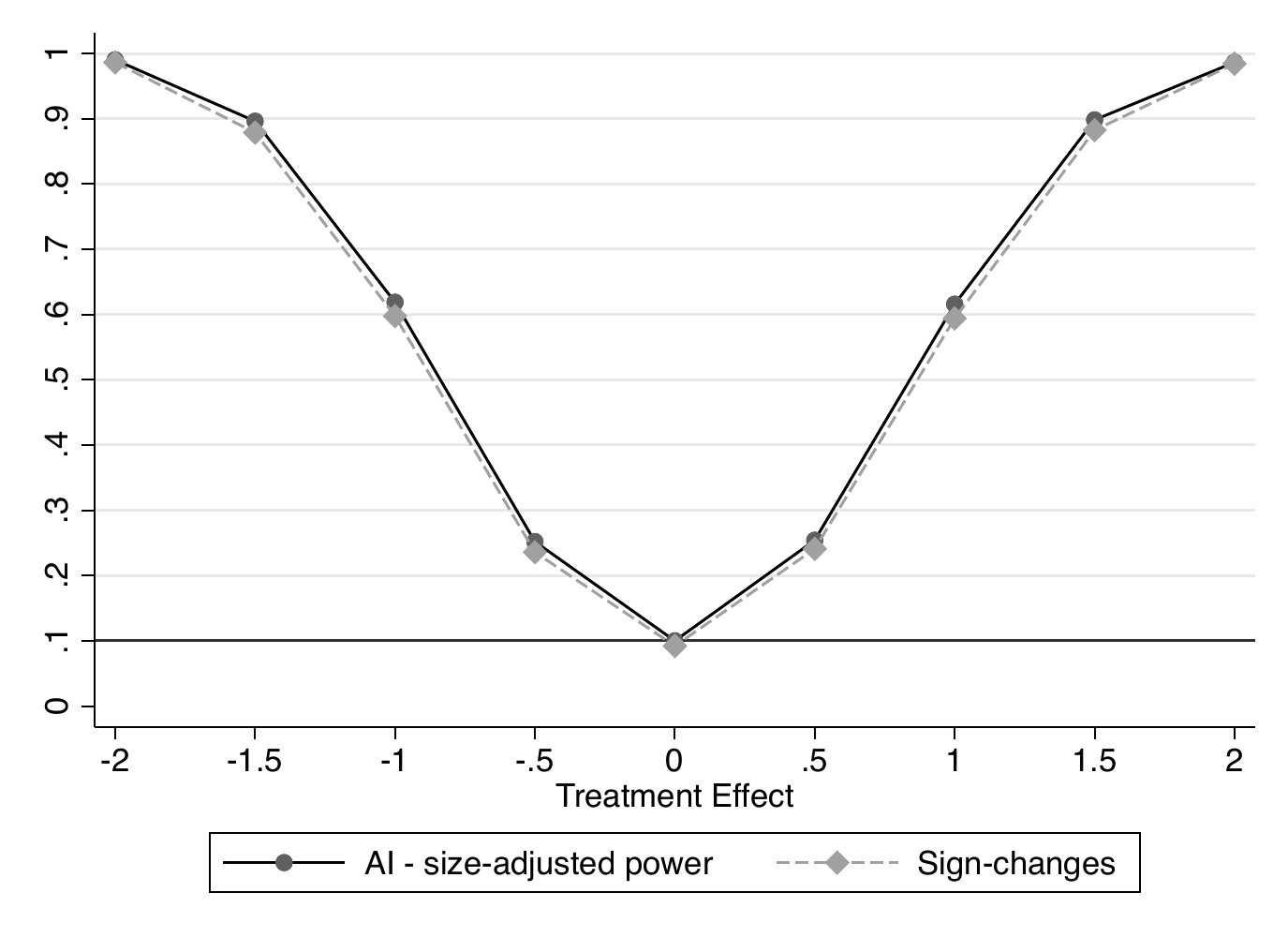} \\

Figure \ref{Figure_MC}.C: $N_1 = 25$ & Figure \ref{Figure_MC}.D:  $N_1 = 50 $ \\

\includegraphics[scale=0.5]{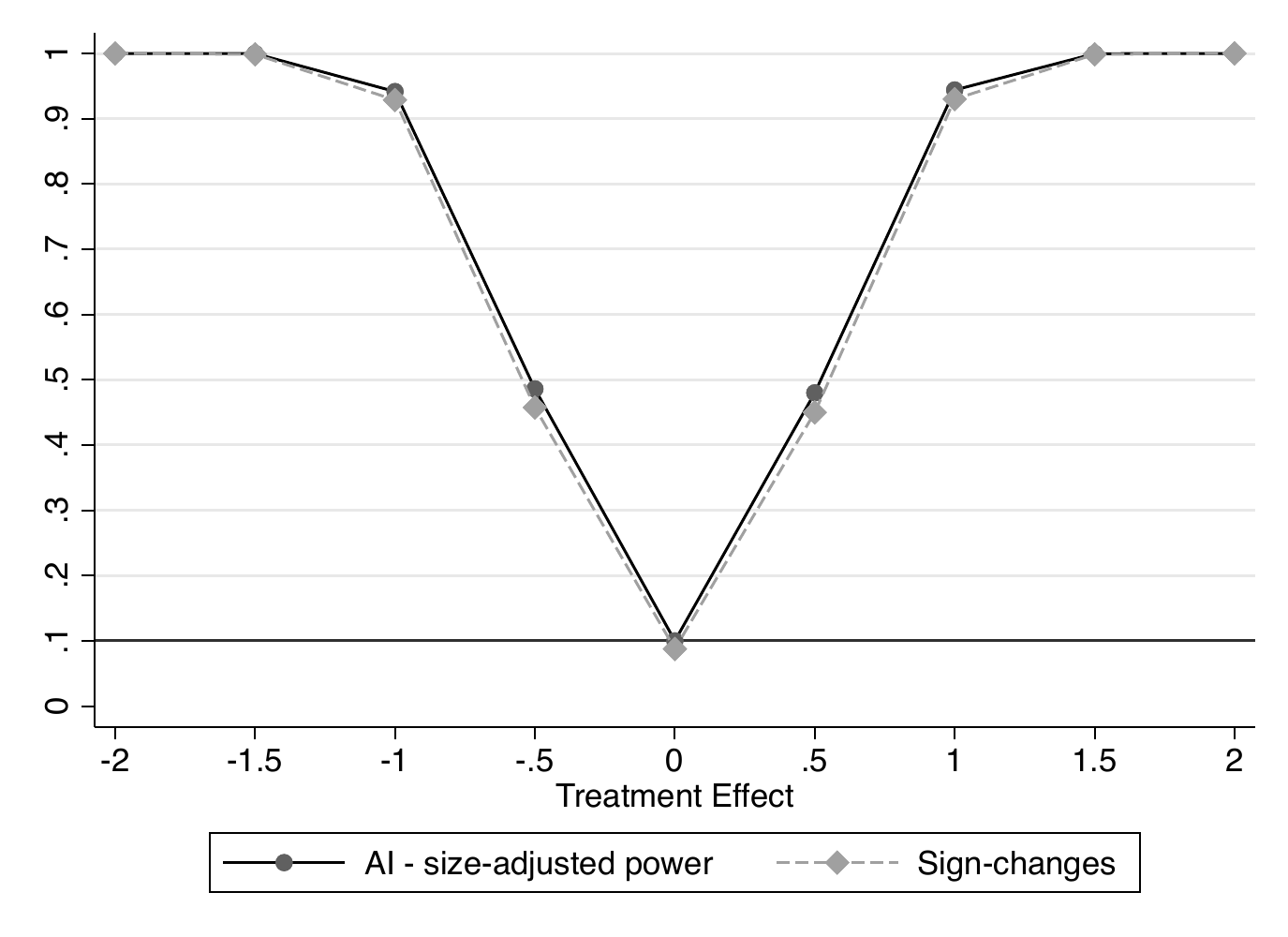} & \includegraphics[scale=0.5]{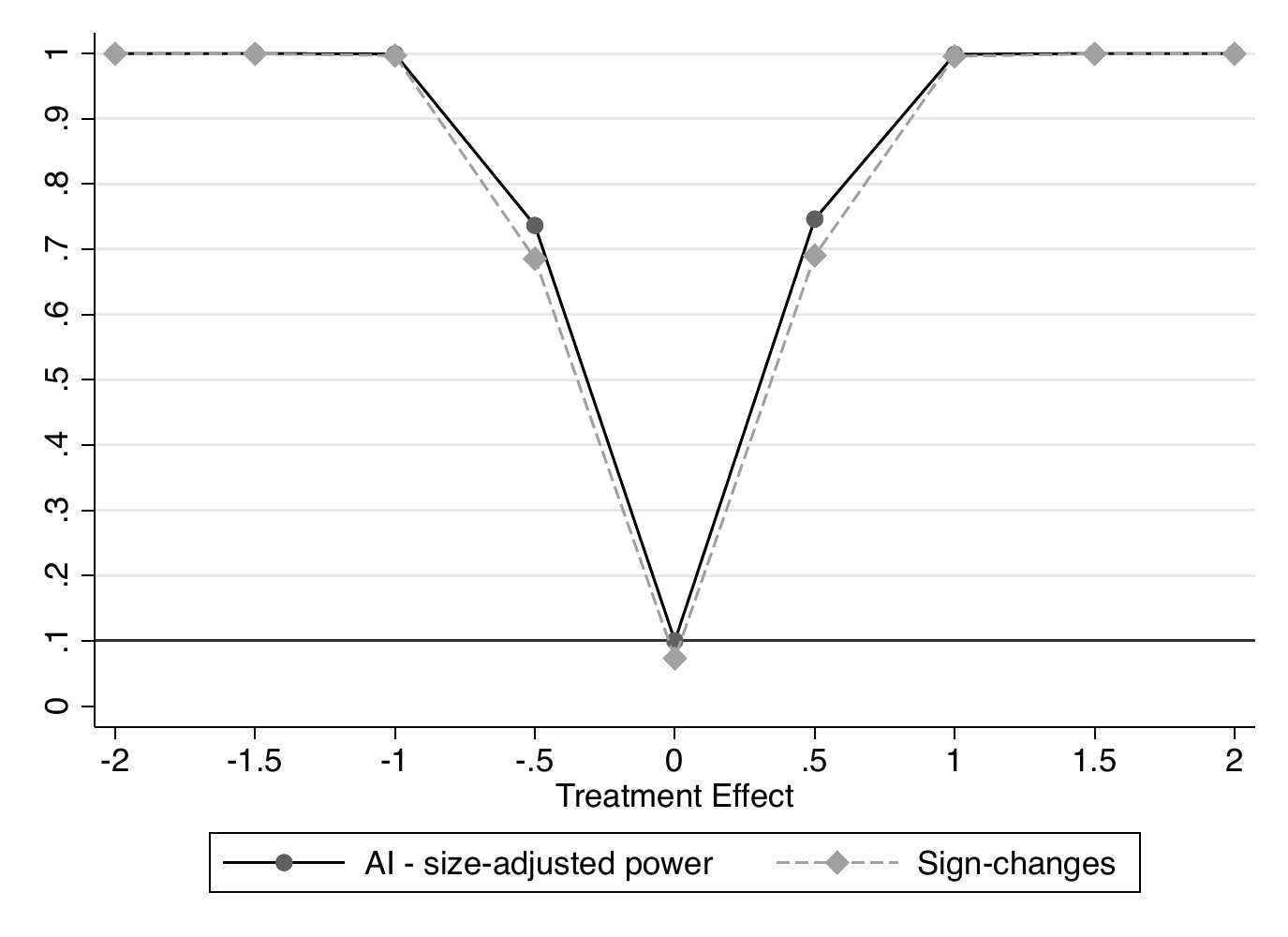} \\

\end{tabular}

\end{center}

\small{Notes: These figures present size-adjusted power for the asymptotic test based on AI at the 10\%-level, and the power for the sign-changes test.  We consider the setting presented in Panel A of Table \ref{Table_MC2},  with $M=4$.  In these simulations,  $X |(W=w) \sim N(0,1)$ for $w \in \{0,1\}$, $Y(0)|(X=x,W=0) \sim N(0,1)$ for all $x \in \mathbb{R}$, $\mu_1(x) = \tau + x$,  $(Y(1)-\mu_1(x))|(X=x,W=1) \sim N(0,1)$ for all $x \in \mathbb{R}$, and $N_0 = 1000$.  For each combination of $N_1$ and $\tau$, we run 5000 simulations. }

\end{figure}

\begin{table}[H]
  \begin{center}
\caption{{\bf Empirical illustration}} \label{empirical_aplication_control_text}
\begin{tabular}{lccccccc}
\hline
\hline
& \multicolumn{3}{c}{Rio de Janeiro} & &  \multicolumn{3}{c}{Sao Paulo} \\ \cline{2-4} \cline{6-8} 
& $M=1$  & $M=4$ & $M=10$ & & $M=1$ & $M=4$ & $M=10$ \\
     & (1) &  (2)   & (3)   & & (4) & (5) & (6)   \\ \hline

\underline{Treatment effects in 2010} \\
 
Point Estimate & 0.087 & -0.003 & 0.046 &  & 0.000 & 0.018 & 0.004 \\
 \\
p-values: \\
AI & 0.091 & 0.941 & 0.086 &  & 0.995 & 0.601 & 0.924 \\
RI-sign changes  & 0.123 & 0.938 & 0.179 &  & 0.996 & 0.609 & 0.917 \\
 \\
 \underline{Treatment effects in 2011} \\
 
 Point Estimate & 0.043 & -0.032 & 0.000 &  & -0.019 & -0.027 & -0.013 \\
 \\
p-values: \\
AI & 0.566 & 0.396 & 0.997 &  & 0.746 & 0.475 & 0.692 \\
RI-sign changes  & 0.662 & 0.438 & 0.997 &  & 0.734 & 0.496 & 0.693 \\
 \\
 
 \underline{Treatment effects in 2012} \\

Point Estimate & 0.070 & -0.019 & 0.006 &  & -0.072 & -0.034 & -0.019 \\
 \\
p-values: \\
AI  & 0.263 & 0.522 & 0.885 &  & 0.169 & 0.383 & 0.616 \\
RI-sign changes  & 0.306 & 0.576 & 0.896 &  & 0.185 & 0.382 & 0.495 \\
 \\

 \hline

\end{tabular}
\end{center}

\footnotesize Note: This table presents non-experimental results using a matching estimator with experimental control schools as treated observations and non-experimental schools as control observations. Columns 1 to 3 present results for Rio de Janeiro using 1, 4, or 10 nearest neighbors in the estimation, while columns 4 to 6 present results for Sao Paulo. We present the estimated effects separately for 2010, 2011, and 2012. For each estimate, we present p-values calculated based on the asymptotic distribution derived by AI, and based on the sign-changes test described in Section \ref{inference}. 

\end{table}

\pagebreak

\appendix

\doublespace

\begin{center}

\huge  Appendix (for online publication)

\end{center}

\section{Proof of Main Results}

\subsection{Lemmas }  \label{}

We start presenting some  lemmas that will be useful in the proofs of the main propositions of the paper. Lemma \ref{Lemma_convergenceX} shows that the covariates of the nearest neighbors to a point $\bar x \in \mathbb{X}_1$ converge in probability to $\bar x$. Lemma \ref{Lemma_convergenceY} considers the distribution of the outcomes of  nearest neighbors to $\bar x$. Finally, Lemmas \ref{Lemma_no_shared} and \ref{Lemma_no_shared2} show that  the probability that two treated observations share the same nearest neighbor converges  to zero. This is valid either when $N_1$ is fixed, or when $N_1 \rightarrow \infty$, but at a sufficiently lower rate than $N_0$.  

\begin{lemma} \label{Lemma_convergenceX}

Suppose Assumptions \ref{sample} and \ref{overlap} hold. Fix $\bar x \in \mathbb{X}_1$, and, for any $M \in \mathbb{N}$,  let $X_{(M)}$ be the covariate of the $M$-closest match of $\bar x$ among the control observations. Then $X_{(M)}  \buildrel p \over \rightarrow \bar x$ when $N_0 \rightarrow \infty$.

\end{lemma}

\begin{proof}
For a given $\epsilon>0$, 
\begin{eqnarray} \nonumber
\mbox{Pr} \left( d(X_{(M)},\bar x) >\epsilon \right) &=& \sum_{m=0}^{M-1} \mbox{Pr} \left( d(X_j,\bar x) \leq \epsilon \mbox{ for exactly $m$  observations  } j \in \mathcal{I}_0  \right) \\
&=&  \sum_{m=0}^{M-1}  \left( \begin{array}{c} N_0 \\ m \end{array} \right) [\mbox{Pr} (d(X_j,\bar x) \leq \epsilon)]^m [\mbox{Pr} (d(X_j,\bar x)>\epsilon)]^{N_0-m}.
\end{eqnarray} 

Since $\bar x \in  \mathbb{X}_1 \subset \mathbb{X}_0$, under Assumption \ref{overlap}, we have that $Pr(d(X_j,\bar x) \leq \epsilon)>0$, which implies that $\mbox{Pr} (d(X_j,\bar x)>\epsilon)<1$. Therefore,  $\mbox{Pr} \left( d(X_{(M)},\bar x)>\epsilon \right) \rightarrow 0$ for any $M \in \mathbb{N}$. 
\end{proof}

\begin{remark} \label{Mahalanobis}
\normalfont
Lemma \ref{Lemma_convergenceX} remains valid if we consider $\tilde d(a,b)$ as the Mahalanobis distance, provided $V = var \left( X | W =0 \right)$ is positive definite. In this case, as $N_0 \rightarrow \infty$ and $N_1$ is fixed, we have that the sample variance/covariance matrix of $X_i$, $\hat V$, converges in probability to $V$.\footnote{With $N_1$ fixed, the covariates of the treated observations will be asymptotically negligible. }  Let $d(X_{(M)},\bar x) = (X_{(M)}-\bar x)'V^{-1}(X_{(M)}-\bar x)$. Then, for any $\epsilon>0$, we can find a constant  $c>0$ such that 
\begin{eqnarray*} 
&&\mbox{Pr} \left(\tilde  d(X_{(M)},\bar x) >\epsilon \right) = \mbox{Pr} \left( d(X_{(M)},\bar x) >\epsilon -  (X_{(M)}-\bar x)'(\hat V^{-1}-V^{-1})(X_{(M)}-\bar x) \right)  \\
&&\leq  \mbox{Pr} \left( d(X_{(M)},\bar x) >\epsilon -  c \norm{\hat V^{-1}-V^{-1}}_2 \right) \\
&&= \mbox{Pr} \left( d(X_{(M)},\bar x) >\epsilon -  c \norm{\hat V^{-1}-V^{-1}}_2 \mid c \norm{\hat V^{-1}-V^{-1}}_2 \geq \frac{\epsilon}{2} \right) \mbox{Pr}\left(c \norm{\hat V^{-1}-V^{-1}}_2 \geq \frac{\epsilon}{2} \right)
\\
&& +\mbox{Pr} \left( d(X_{(M)},\bar x) >\epsilon -  c \norm{\hat V^{-1}-V^{-1}}_2 \mid c \norm{\hat V^{-1}-V^{-1}}_2<\frac{\epsilon}{2} \right) \mbox{Pr}\left(c \norm{\hat V^{-1}-V^{-1}}_2<\frac{\epsilon}{2} \right) \\
&&\leq \mbox{Pr} \left(c \norm{\hat V^{-1}-V^{-1}}_2 \geq \frac{\epsilon}{2} \right) + \mbox{Pr} \left( d(X_{(M)},\bar x) >\frac{\epsilon}{2} \right),
\end{eqnarray*}
where $\mbox{Pr} \left(c \norm{\hat V^{-1}-V^{-1}}_2 \geq \frac{\epsilon}{2} \right) \rightarrow 0$ from consistency of $\hat V$, while $ \mbox{Pr} \left( d(X_{(M)},\bar x) >\frac{\epsilon}{2} \right) \rightarrow 0$ from Lemma \ref{Lemma_convergenceX}. The first inequality follows from $\mathbb{X}_0$ compact.

\end{remark}

\begin{lemma}
\label{Lemma_convergenceY}
Suppose Assumptions \ref{sample}, \ref{CIA},  \ref{overlap}, and \ref{assumption_mu}(b) hold. Fix $\bar x \in \mathbb{X}_1$, and, for any $M \in \mathbb{N}$,  let $Y_{(M)}$ be the outcome of the $M$-closest match of $\bar x$ among the control observations $i \in \mathcal{I}_0$. Then $Y_{(M)}  \buildrel d \over \rightarrow Y(0) | (X=\bar x)$ when $N_0 \rightarrow \infty$. 

\end{lemma}

\begin{proof}
 Let  $h(y)$ be a  continuous and bounded function, and let $\tilde h(x) = \mathbb{E}[h(Y(0))|X=x]$.\footnote{Note that, under Assumption \ref{CIA}, we have that  $ \mathbb{E}[h(Y(0))|X=x] =  \mathbb{E}[h(Y(0))|X=x,W=0]$.} From Assumption \ref{assumption_mu}(b), we have that  $\tilde h(x)$ is continuous and bounded as a function of $x$. Now for  a given $\bar x \in \mathbb{X}_1$, let $X_{(M)}$ be the covariate of the $M-$nearest neighbor to $\bar x$  among the control observations among the control observations $i \in \mathcal{I}_0$. Then we have that
\begin{eqnarray} 
\mathbb{E}[h(Y_{(M)})]&=&\mathbb{E} \left\{\mathbb{E}[h(Y_{(M)}) | X_{(M)}] \right\} = \mathbb{E} \left\{\tilde h(X_{(M)}) \right\} \rightarrow \tilde h( \bar x) = \mathbb{E}  [ h(Y(0)) | X= \bar x],
\end{eqnarray}
where convergence follows from Lemma \ref{Lemma_convergenceX} and from the fact that $\tilde h(x)$ is continuous and bounded from Assumption  \ref{assumption_mu}(b). By the Portmanteau Lemma, we have that $Y_{(M)} \buildrel d \over \rightarrow Y(0)|\{X = \bar x\}$ for any $M \in \mathbb{N}$. 
\end{proof}

\begin{lemma} \label{Lemma_no_shared}

Let $\Omega$ be the event that there is no shared nearest neighbor. Suppose Assumptions \ref{sample} and \ref{overlap} hold. Assume either that (i) $N_1$ is fixed or that (ii) there are $0 < \theta<\infty$ and $r>1$ such that  $N_1^r/N_0 \rightarrow \theta$. Then $Pr(\Omega) \rightarrow 1$ when $N_0 \rightarrow \infty$. 

\end{lemma}

\begin{proof}
If there are two distinct treated units, $i,j \in \mathcal{I}_1$ that share the same nearest neighbor, then it must be that $\norm{X_i - X_j} \leq \norm{X_i - X_{(M)}^i } + \norm{X_j - X_{(M)}^j }$, where $X_{(M)}^i $ is the covariate of the $M-$th nearest neighbor of observation $i$ among the control observations. Therefore, the probability of having shared nearest neighbors is bounded by
\begin{eqnarray}
Pr(\Omega^c) \leq Pr\left( \underset{i,j \in \mathcal{I}_1}{\mbox{min}} \norm{X_i - X_j}  < 2 \underset{i \in \mathcal{I}_1}{\mbox{max}} \norm{X_i - X^i_{(M)}}   \right),
\end{eqnarray}
where, with some abuse of notation, $ \underset{i,j \in \mathcal{I}_1}{\mbox{min}}$ should always be understood as $i \neq j$. Let $\tilde r \in (1,r)$. If $N_1$ is fixed, then just substitute  $\tilde r/r$ for a constant in $(0,1)$. Then 
\begin{eqnarray} \label{prob_omega}
Pr(\Omega^c) &\leq&  Pr\left(  \underset{i,j \in \mathcal{I}_1}{\mbox{min}} \norm{X_i - X_j} \leq 2 N_0^{-\tilde r/rk}   \right) \\ \nonumber
&&  + Pr\left(  \underset{i \in \mathcal{I}_1}{\mbox{max}} \norm{X_i - X^i_{(M)}} > N_0^{-\tilde r/rk} \middle| \underset{i,j \in \mathcal{I}_1}{\mbox{min}} \norm{X_i - X_j} >2 N_0^{-\tilde r/rk}   \right) .
\end{eqnarray}

We show that the two terms on the right hand side of equation (\ref{prob_omega}) converge to zero. Let $\Gamma_s$ be the event in which $X_s$ is one of the observations that minimizes $ \norm{X_i - X_j} $.    For the first term, note that  
\begin{eqnarray} \label{eq1234} \nonumber
Pr\left(  \underset{i,j \in \mathcal{I}_1}{\mbox{min}} \norm{X_i - X_j} \leq 2 N_0^{-\tilde r/rk}   \right) &\leq& \sum_{s \in \mathcal{I}_1} Pr\left(  \underset{i,j \in \mathcal{I}_1}{\mbox{min}} \norm{X_i - X_j} \leq 2 N_0^{-\tilde r/rk}   \middle| \Gamma_s \right)Pr\left(\Gamma_s\right) \\
&\leq& \frac{2}{N_1} \sum_{s \in \mathcal{I}_1}  \underset{a \in \mathbb{X}_1}{\mbox{max}}  Pr\left( \norm{Z_{(1)}-a} \leq 2 N_0^{-\tilde r/rk}   \right)  \\
&=& 2\underset{a \in \mathbb{X}_1}{\mbox{max}}  Pr\left(  \norm{Z_{(1)}-a} \leq  2 N_0^{-\tilde r/rk}   \right),
\end{eqnarray}
where $Z_{(1)}$ is the nearest neighbor to $a$ when we sample $N_1$ observations of $X_i$ with density $f_1(x)$. We used that the probability that $i$ is one of the pair with minimum distance is $2/N_1$. Moreover, $Pr\left(  \underset{i,j \in \mathcal{I}_1}{\mbox{min}} \norm{X_i - X_j} \leq 2 N_0^{-\tilde r/rk}   \middle| \Gamma_s \right)$ is the integral of $Pr\left( \norm{Z_{(1)}-a} \leq 2 N_0^{-\tilde r/rk}   \right)$ where we integrate over $a$ using the distribution of $X_w$ conditional on $\Gamma_s$. Therefore, this probability is bounded by the maximum over $a \in \mathbb{X}_1$.

We show that this bound converges to zero.  We use some results and definitions used in the proof of  Lemma 1 from AI. Let $\mathbb{S}_k = \{\omega \in \mathbb{R}^k : \norm{\omega}=1\}$ be the unit $k$ sphere and let $\lambda_{\mathbb{S}_k}$ be its surface measure. Fix an $a \in \mathbb{X}_1$, and let $f_{Z_{(1)}}(z;a)$ be the density of $Z_{(1)}$. A derivation of $f_{Z_{(1)}}(z;a)$ can be found in the proof of Lemma 1 from AI. With a slight abuse of notation, we define throughout $\tilde c$ as a constant that may vary across equations, but that does not depend on $N_1$, $N_0$, and $a$. Therefore, for any $a \in \mathbb{X}_1$,
\begin{eqnarray}
&& Pr\left(  \norm{Z_{(1)} -a} \leq 2  N_0^{-\tilde r/rk}  \right) = \int_0^{2N_0^{-\tilde r/rk}} q^{k-1} \left( \int_{\mathbb{S}_k}f_{Z_{(1)}}\left( a + r \omega;a \right) \lambda_{\mathbb{S}_k}(d \omega) \right)dq \\ \nonumber
&&\leq  \int_0^{2N_0^{-\tilde r/rk}} q^{k-1} \left( \int_{\mathbb{S}_k} N_1 f_1\left( a + r \omega \right) \left( 1 - Pr(\norm{X-a} \leq  \norm{r \omega} | W=1) \right)^{N_1 - 1} \lambda_{\mathbb{S}_k}(d \omega) \right)dq \\ \nonumber
&&\leq \tilde c N_1  \int_0^{2N_0^{-\tilde r/rk}} q^{k-1} dq = \tilde c N_1 N_0^{-\tilde r/r},
\end{eqnarray}
where $\tilde c$ does not depend on $a$, $N_0$ and $N_1$. If $N_1$ is fixed, then $N_1 N_0^{-\tilde r/r} \rightarrow 0$ uniformly in $a$. If $N_1$ diverges, then  $N_1 N_0^{-\tilde r/r} =  N_0^{(1-\tilde r)/r}(N_1^r/N_0)^{1/r}$, which also converges uniformly to zero because  $\tilde r>1$. Therefore, the first term  on the right hand side of equation (\ref{prob_omega}) converges to zero.

We now consider the second term   on the right hand side of equation (\ref{prob_omega}). Fix a sequence $\{X_1,...,X_{N_1}\} = \{a_1,...,a_{N_1}\}$. Then
\begin{eqnarray} 
Pr\left(  \underset{i \in \mathcal{I}_1}{\mbox{max}} \norm{a_i - X^i_{(M)}} >  N_0^{-\tilde r/rk}   \right)  &\leq& \sum_{i \in \mathcal{I}_1} Pr\left(   \norm{a_i - X^i_{(M)}} >  N_0^{-\tilde r/rk}  \right) \\
& \leq& N_1  \underset{a \in \mathbb{X}_1}{\mbox{max}}  Pr\left(   \norm{a - X_{(M)}} >  N_0^{-\tilde r/rk}   \right),
\end{eqnarray}
where $X_{(M)}$ is the $M-$th closest match to $a$ among the control observations.

Let $\bar D$ be the diameter of $\mathbb{X}_0$, and fix an  $a \in \mathbb{X}_1$. Since $\mathbb{X}_0$ is compact, $\bar D<\infty$. Let $f_{X_{(M)}}(x;a)$ be the density of  $X_{(M)}$. Then,
\begin{eqnarray*} 
&& Pr\left(   \norm{a - X_{(M)}} >  N_0^{-\tilde r/rk}   \right) = \int_{N_0^{-\tilde r/rk}}^{\bar D} q^{k-1} \left(  \int_{\mathbb{S}_k}f_{X_{(M)}}(a + rw;a) \lambda_{\mathbb{S}_k} (d \omega)  \right)dq \\
&=& \nonumber \int_{N_0^{-\tilde r/rk}}^{\bar D} q^{k-1} \left(   \begin{array}{l}  \int_{\mathbb{S}_k}N_0 f_0(a+rw) \left( \begin{array}{l}  N_0-1 \\ M-1 \end{array}  \right) \left( 1 - Pr(\norm{X-a} \leq \norm{rw} \mid W=0)   \right)^{N_0-M} \times  \\   \times Pr(\norm{X-a} \leq \norm{rw} | W=0)^{M-1}  \lambda_{\mathbb{S}_k} (d \omega)    \end{array} \right)dq \\
&\leq& \tilde c N_0 \left( \begin{array}{l}  N_0-1 \\ M-1 \end{array}  \right)\left( 1 - Pr(\norm{X-a} \leq N_0^{-\tilde r/rk} | W=0) \right)^{N_0-M} \int_{N_0^{-\tilde r/rk}}^{\bar D} q^{k-1} dq \\
&=& \tilde c N_0  (N_0-1)...(N_0-M+1) \left( 1 - Pr(\norm{X-a} \leq N_0^{-\tilde r/rk}|W=0)\right)^{N_0}(1+o(1)) (\bar D^k - N_0^{-\tilde r /r}) \\
&=&  (1+o(1)) h(N_0) \left( 1 - Pr(\norm{X-a} \leq N_0^{-\tilde r/rk} | W=0) \right)^{N_0},
\end{eqnarray*}
for some polynomial $h(N_0)$.

Now note that, since $f_0(x)$ is bounded from below in $\mathbb{X}_1$, and for all points $x \in \mathbb{X}_1$ at least a fraction $\phi$ of any sphere around $x$ belongs to  $\mathbb{X}_1$, we have that 
\begin{eqnarray*}
Pr(\norm{X-a} \leq N_0^{-\tilde r/rk} | W=0) &=& \int_0^{N_0^{-\tilde r/rk}} q^{k-1} \left(  \int_{\mathbb{S}_k}f_{0}(a + rw) \lambda_{\mathbb{S}_k} (d \omega)  \right)dq \geq   c  N_0^{-\tilde r/r},
\end{eqnarray*}
for some positive constant $c$.

Therefore,
\begin{eqnarray*} 
N_1  \underset{a \in \mathbb{X}_1}{\mbox{max}}  Pr\left(   \norm{a - X_{(M)}} >  N_0^{-\tilde r/rk}   \right) &\leq&  (1+o(1)) N_1 h(N_0) \left( 1 - cN_0^{-\tilde r/r} \right)^{N_0}\\
 &\leq& N_1  (1+o(1)) h(N_0) exp \left(  - cN_0^{1 - \tilde r/r} \right) \rightarrow 0,
\end{eqnarray*}
since $\tilde r < r$. Combining these results, we have that $Pr(\Omega) \rightarrow 1$.
\end{proof}

\begin{lemma} \label{Lemma_no_shared2}

Let $S_{N_1} =1$ if there is a shared nearest neighbor, and zero otherwise. Suppose Assumptions \ref{sample} and \ref{overlap} hold. Assume that there are $0 < \theta<\infty$ and $r>2$ such that  $N_1^r/N_0 \rightarrow \theta$. Then $S_{N_1} \buildrel a.s. \over \rightarrow 0$.

\end{lemma}

\begin{proof}
For any $u \in (0,1)$, note that $Pr(|S_{N_1} - 0|>u ) = Pr(\Omega^c)$. Now set $\tilde r \in (2,r)$. From the proof of Lemma \ref{Lemma_no_shared}, we have that 
\begin{eqnarray} 
Pr(|S_{N_1} - 0|>u)  \leq  (1+o(1)) h(N_1) exp \left(  - c(1+o(1))N_1^{r - \tilde r} \right) + \tilde c (1+o(1)) N_1^{1- \tilde r}.
\end{eqnarray}

Therefore, $\sum_{N_1=1}^\infty Pr(|S_{N_1} - 0|>u) <\infty$, implying that $S_{N_1} \buildrel a.s. \over \rightarrow 0$.
\end{proof}

\subsection{Proof of Proposition \ref{unbiased} }  \label{proof_unbiased}

\begin{proposition_b}{3.1:}
(1) Under Assumptions \ref{sample}, \ref{CIA},  \ref{overlap}, and \ref{assumption_mu}(a), $\mathbb{E}[\hat \tau ] \rightarrow \tau$ when $N_0 \rightarrow \infty$ and $N_1$ is fixed.

(2) Under Assumptions \ref{sample}, \ref{CIA},  \ref{overlap}, and \ref{assumption_mu}(b), 
\begin{eqnarray} \nonumber
\hat \tau \buildrel d \over \rightarrow  \tau + \frac{1}{N_1}  \sum_{i \in \mathcal{I}_1} \kappa_i  \mbox{ when $N_0 \rightarrow \infty$ and $N_1$ is fixed}, 
\end{eqnarray}
where the CDF of $\kappa_i$ is given by $\widetilde G(\kappa) = \int_{x \in \mathbb{X}_1} G(\kappa; x)f_1(x)dx$, and $\mathbb{E}[\kappa_i]=0$.   Moreover, $\{\kappa_i\}_{i \in \mathcal{I}_1}$ is mutually independent.

\end{proposition_b}

\begin{proof}

\textbf{Part 1: }
Without loss of generality, let $i=1,...,N_1$ be the treated observations, and let  $\mathbf{X} = [X_1 ~ X_2 ~ \cdots ~ X_{N_1}]$. Fix $ \bar{\mathbf{x}} =  [\bar x_1 ~ \bar x_2 ~ \cdots ~ \bar x_{N_1}] \in \otimes_{i=1}^{N_1} \mathbb{X}_1 $. We have that
\begin{eqnarray}
\mathbb{E}[\hat \tau | \mathbf{X} = \bar{\mathbf{x}}  ] = \frac{1}{N_1}  \sum_{i \in \mathcal{I}_1}  \left( \mu_1(\bar x_i)  - \mathbb{E} \left[ \frac{1}{M} \sum_{m=1}^M \mu_0 ( X_{(m)}^i) \mid X_i = \bar x_i  \right] \right). 
\end{eqnarray}

Since $\mu_0(x)$ is continuous and $\mathbb{X}_0$ is bounded,  and $X_{(m)}^i \buildrel p \over \rightarrow X_i = \bar x_i$ (Lemma \ref{Lemma_convergenceX}), it follows that  $\mathbb{E}[\mu_0 ( X_{(m)}^i) | X_i = \bar x_i] \rightarrow \mu_0(\bar x_i)$. Therefore, 
\begin{eqnarray} \label{eq_unbiased}
\mathbb{E}[\hat \tau | \mathbf{X} =  \bar{\mathbf{x}}] \rightarrow \frac{1}{N_1}  \sum_{i \in \mathcal{I}_1}  \left( \mu_1(\bar x_i) - \mu_0(\bar x_i) \right). 
\end{eqnarray}

Now we consider the unconditional expectation of $\hat \tau$,
\begin{eqnarray}
\mathbb{E}[\hat \tau  ] =\mathbb{E} \{ \mathbb{E}[\hat \tau | \mathbf{X}  ]   \}=  \frac{1}{N_1}  \sum_{i \in \mathcal{I}_1} \mathbb{E} \left[ \mu_1(X_i)  -  \frac{1}{M} \sum_{m=1}^M \mu_0 ( X_{(m)}^i)  \right].
\end{eqnarray}

We need that, for $i \in \mathcal{I}_1$, $\mathbb{E}[\mu_0 ( X_{(m)}^i) ] \rightarrow \mathbb{E}[\mu_0 ( X) \mid W =1 ] $.  We know that $\mathbb{E}[\mu_0 ( X_{(m)}^i) | X_i = \bar x_i ] \rightarrow \mu_0(\bar x_i)$ for all $\bar x_i \in \mathbb{X}_1$. Using the fact that $\mu_0(x)$ is continuous and $\mathbb{X}_0$ is compact, we have that $\mathbb{E}[\mu_0 ( X_{(m)}^i)  ] =\mathbb{E}\{ \mathbb{E}[\mu_0 ( X_{(m)}^i) | X_i] \} \rightarrow \mathbb{E} [\mu_0(X_i) ] = \mathbb{E} [\mu_0(X) \mid W=1]$. Therefore,
\begin{eqnarray}
\mathbb{E}[\hat \tau ] \rightarrow \mathbb{E} \left[ \mu_1(X)  - \mu_0 ( X) \mid W=1 \right],
\end{eqnarray}
 which proves part 1 of Proposition \ref{unbiased}.

\textbf{Part 2: }  For each $x \in \mathbb{X}_1$, let $\xi_x \sim (Y(1) - \mu_1(x)) | (X=x,W=1)$ and  $\eta_x  \sim (Y(0) - \mu_0(x)) | (X=x,W=0)$. Moreover, let $G(\kappa; x)$ be the CDF of $(\mu_1(x) - \mu_0(x) - \tau) + \xi_x - \frac{1}{M} \sum_{m=1}^M \eta_x^m$, where $\{\eta_x^m \}_{m=1}^M$ are iid copies of $\eta_x$, and  $(\xi_x,\eta^1_x,\cdots,\eta^M_x)$ is mutually  independent. 

For a given $i \in \mathcal{I}_1$, consider $ Y_i - \frac{1}{M} \sum_{m=1}^M Y^i_{(m)}$. If we condition on $X_i = \bar x$, then  $(Y_i | X_i = \bar x) \sim (\mu_1(\bar x) + \xi_{\bar x})$. From  Lemma \ref{Lemma_convergenceY}, we also know that, conditional on $X_i = \bar x$, $Y^i_{(m)}$ converges in distribution to $(\mu_0(\bar x) + \eta^m_{\bar x})$.  Moreover, given Assumption \ref{sample}, conditional on $X_i = \bar x$, the asymptotic distribution of $(Y_i,Y^i_{(1)},...,Y^i_{(M)})$ is mutually independent.  

Combining these results, conditional on $X_i = \bar x$,  $ Y_i - \frac{1}{M} \sum_{m=1}^M Y^i_{(m)}$ converges in distribution to $\tau + (\mu_1(\bar x) - \mu_0(\bar x) - \tau ) + \xi_{\bar x} - \frac{1}{M} \sum_{m=1}^M \eta^m_{\bar x}$, where $(\xi_{\bar x},\eta_{\bar x}^1, \cdots ,\eta_{\bar x}^M)$ is mutually independent. Therefore, for all $i \in \mathcal{I}_1$,
\begin{eqnarray}
Pr \left( Y_i - \frac{1}{M} \sum_{m=1}^M Y^i_{(m)} - \tau \leq c \mid X_i = \bar x \right) \rightarrow G(c ; \bar x),
\end{eqnarray}
for all $c \in \mathbb{R}$ in which $G(c ; \bar x)$ is continuous.

Now we show that  $Pr\left(Y_i - \frac{1}{M} \sum_{m=1}^M Y^i_{(m)} - \tau \leq c \right) \rightarrow \int_{x \in \mathbb{X}_1} G(c; x)f_1(x)dx \equiv \widetilde G(c)$ for all $c \in \mathbb{R}$ in which $\tilde G(c)$ is continuous. 

Define $\Sigma(c) = \{x \in \mathbb{X}_1 : G(c; x) \mbox{ is continuous at } c \}$.  Since $G(c; x)$ are CDF's, we have that  $G(c; x) -  \lim_{h \uparrow c}G(h;x)=0$ if $x \in \Sigma(c)$, and  $G(c; x) -  \lim_{h \uparrow c}G(h;x)>0$ if $x \notin \Sigma(c)$. Now let $c \in \mathbb{R}$ be such that $\widetilde G(c)$ is continuous. Then
\begin{eqnarray*}
0 &=& \widetilde G(c) -  \lim_{h \uparrow c} \widetilde G(h) = \int_{x \in \mathbb{X}_1} G(c; x)f_1(x)dx -   \lim_{h \uparrow c} \int_{x \in \mathbb{X}_1} G(h; x)f_1(x)dx \\
&=& \int_{x \in \mathbb{X}_1} \left[G(c; x) -  \lim_{h \uparrow c} G(h; x) \right] f_1(x)dx =  \mathbb{E}\left[G(c;X) - \lim_{h \uparrow c}G(h;X) \mid W=1 \right].
\end{eqnarray*}

Since $G(c;X) - \lim_{h \uparrow c}G(h;X) \geq 0$, it follows that    $Pr(X \in \Sigma(c) | W=1)=1$ if $\widetilde{G}(\cdot)$ is continuous at $c$.

Therefore, for any $c \in \mathbb{R}$ in which $\widetilde G(c)$ is continuous, we have
\begin{eqnarray*}
Pr \left( Y_i - \frac{1}{M} \sum_{m=1}^M Y^i_{(m)} - \tau \leq c  \right) &=& \int_{x \in \mathbb{X}_1} Pr \left( Y_i - \frac{1}{M} \sum_{m=1}^M Y^i_{(m)} - \tau \leq c \mid X_i =  x \right) f_1(x)dx \\
&\rightarrow& \int_{x \in \mathbb{X}_1} G(c;x) f_1(x)dx = \widetilde G(c),
\end{eqnarray*}
where this convergence follows from the fact that  $Pr \left( Y_i - \frac{1}{M} \sum_{m=1}^M Y^i_{(m)} - \tau \leq c \mid X_i =  x \right)  \rightarrow G(c;x)$ for almost all $x \in \mathbb{X}_1$, and $Pr()$ is bounded. 

Therefore, unconditionally, $ Y_i - \frac{1}{M} \sum_{m=1}^M Y^i_{(m)}$ converges in distribution to $\tau + \kappa_i$, where the CDF of $\kappa_i$ is given by $\widetilde G(\kappa)$. From the law of iterated expectations and from the fact that $\tau = \mathbb{E}[\mu_1(X)-\mu_0(X) | W=1]$, we have that $\mathbb{E}[\kappa_i]=0$.

Finally, note that $(\kappa_1,...,\kappa_{N_1})$ is mutually independent given Assumption \ref{sample} and from the fact that the probability that two treated observations share the same nearest neighbor converges to zero (Lemma \ref{Lemma_no_shared}).
\end{proof}

\subsection{Proof of Proposition \ref{test}} \label{proof_sign_changes}

\begin{proposition_b}{4.1:}
Suppose Assumptions \ref{sample}, \ref{CIA},  \ref{overlap},  \ref{assumption_mu}(b), and  \ref{Assumption_symmetry} hold. Assume also that the distribution of $Y$ is continuous. If we consider the problem of testing $H_0:  \tau =c$, then, for any $\alpha \in (0,1)$, $\mbox{limsup}_{N_0 \rightarrow \infty} \mathbb{E} \left[  \phi( S_{N_0}) \right]  \leq \alpha$ when $N_1$ is fixed. 

\end{proposition_b}

\begin{proof}
From Proposition \ref{unbiased}, if we consider $i \in \mathcal{I}_1$ as the only treated observation, we have that 
\begin{eqnarray}
\hat \tau_i^{N_0} =  Y_i - \frac{1}{M} \sum_{m=1}^M Y_{(m)}^i - c    \buildrel d \over \rightarrow  (\tau - c)  + \kappa_i.
\end{eqnarray}

Moreover, given Assumption \ref{Assumption_symmetry}, $\kappa_i$ is symmetric around zero.  Therefore, under the null $\tau = c$, the limiting distribution of $\hat \tau_i^{N_0}$ is symmetric around zero.  Also, since the probability of different treated observations having shared nearest neighbors converges to zero (Lemma \ref{Lemma_no_shared}), we have that, under the null,  $S_{N_0} \buildrel d \over \rightarrow (\kappa_1,...,\kappa_{N_1})$, where $\kappa_i$ is mutually independent across $i$. Therefore, the limiting distribution of $S_{N_0}$, under the null, is invariant  to the transformations in ${\textbf{G}}$.

We also have that the test statistic function $ T(S)$ is continuous. Finally, we show that, for two distinct elements $g \in  G$ and $g' \in  G$, either $T( gS ) = T( g'S)$ for all possible realizations of $S$, or $Pr(T( gS) \neq T(g' S)) = 1$.
If $g$ and $g'$ are such that $g_i = g_i'$ for all $i$, or $g_i = - g_i'$ for all $i$, then $T(g S ) = T(g' S)$ for all possible realizations of $S$. Otherwise, given that $S$ is a continuous random variable, $Pr(T(g S ) \neq T(g' S)) = 1$.  Therefore, we can apply  Theorem 3.1 from  \cite{Canay}.
\end{proof}

\subsection{Proof of Proposition \ref{sign_changes_asympt}} \label{proof_sign_changes_asympt}

Before we proceed with the proof of Proposition \ref{sign_changes_asympt}, we first show that Assumption \ref{overlap} implies that for  all $x \in \mathbb{X}_1$, $Pr(W = 1 | X=x) < 1-\eta$ for some $\eta>0$. This is important so that we can apply Corollary 1 from AI.

\begin{lemma}

Assumption \ref{overlap} implies that there is a $\eta>0$ such that, for all, $x \in \mathbb{X}_1$, $Pr(W = 1 | X=x) < 1-\eta$.

\end{lemma}

\begin{proof}
Fix $x \in \mathbb{X}_1$, and let $e>0$. Then\footnote{Note that $\mbox{Pr}(W=1)$ in this case refers the the probability in the superpopulation $(Y,X,W)$.}
\begin{eqnarray} \nonumber
Pr(W = 1 | d(X,x)\leq e) &=& \frac{Pr(d(X,x)\leq e | W = 1 )Pr(W=1)}{Pr(d(X,x)\leq e | W = 0 )Pr(W=0)+Pr(d(X,x)\leq e | W = 1 )Pr(W=1)} \\
&=& \frac{1}{1 + \frac{Pr(d(X,x)\leq e | W = 0 )Pr(W=0)}{Pr(d(X,x)\leq e | W = 1 )Pr(W=1)} }.
\end{eqnarray}

Now given Assumption \ref{overlap}, there are positive constants $C_0$ and $C_1$ such that 
\begin{eqnarray}
Pr(d(X,x)\leq e | W = 0 ) \geq C_0 e^k \mbox{ and } Pr(d(X,x)\leq e | W = 1 ) \leq C_1 e^k,
\end{eqnarray}
implying that 
\begin{eqnarray}
\frac{Pr(d(X,x)\leq e | W = 0 )}{Pr(d(X,x)\leq e | W = 1 )} \geq C,
\end{eqnarray}
for some constant $C>0$, regardless of $e$. Taking the limit when $e\rightarrow 0$, we have that $Pr(W = 1 | d(X,x)\leq e) \leq 1 - \eta$ for some $\eta>0$.
\end{proof}

\begin{proposition_b}{4.2:}

Suppose Assumptions \ref{sample}, \ref{CIA}, \ref{overlap},  \ref{Assumption_rates}, and \ref{Assumption_appendix2} hold.  If we consider the problem of testing $H_0: \tau =c$, then the sign-changes test is asymptotically valid when  $N_1,N_0 \rightarrow \infty$.

\end{proposition_b}

\begin{proof}
For $i \in \mathcal{I}_1$, let $\hat \tau_i^{N_0} = \left(Y_i - \frac{1}{M}  \sum_{j \in \mathcal{J}_M(i)} Y_j \right)- c$. Consider the test statistic
\begin{eqnarray}
\sqrt{N_1} \widehat T = \frac{ \frac{1}{\sqrt{N_1}} \sum_{i \in \mathcal{I}_1}  \hat \tau_i^{N_0} }{\sqrt{\frac{1}{N_1 - 1} \sum_{i \in \mathcal{I}_1}  (\hat \tau_i^{N_0} )^2 - \frac{N_1}{N_1 - 1} \left( \frac{1}{N_1} \sum_{i \in \mathcal{I}_1}  \hat \tau_i^{N_0} \right)^2  }},
\end{eqnarray}
and its  counterpart using the sign-changes transformation 
\begin{eqnarray}
\sqrt{N_1} \widehat T^\ast = \frac{ \frac{1}{\sqrt{N_1}} \sum_{i \in \mathcal{I}_1} g_i \hat \tau_i^{N_0} }{\sqrt{\frac{1}{N_1 - 1} \sum_{i \in \mathcal{I}_1}  (\hat \tau_i^{N_0} )^2 - \frac{N_1}{N_1 - 1} \left( \frac{1}{N_1} \sum_{i \in \mathcal{I}_1} g_i \hat \tau_i^{N_0} \right)^2  }},
\end{eqnarray}
where $g_i$ equals $1$ with probability 1/2 and $ -1$ with probability 1/2, and they are independent across $i$, and independent of $\hat \tau_i^{N_0}$. 

Let $Pr^\ast$ denote the probability measure induced by the sign changes, and $\mathbb{E}^\ast$ the corresponding expectation conditional on a realization of the random variables $X_i$ and $\epsilon_i$. Define $\widehat F_{T^\ast}(t) \equiv Pr^\ast \left( \sqrt{N_1} \widehat T^\ast  \leq t \right) $. Note that, given the realizations of $X_i$ and $\epsilon_i$, we can derive $\widehat F_{T^\ast}(t) $ by considering all transformations $g \in \mathbf{G}$ with equal probabilities, or approximate that as well as we want by drawing bootstrap samples of $g \in \mathbf{G}$.  Also, define  $F_{T}(t) \equiv Pr\left( \sqrt{N_1} \widehat T \leq t \right)$ and $\Phi(t)$ as the CDF of a standard normal. We want to show that, under the null, $\mbox{sup}_{t \in \mathbb{R}} \left| \widehat F_{T^\ast}(t) -  F_{T}(t)  \right| \buildrel p \over \rightarrow 0$. Note that 
\begin{eqnarray} \label{inequality}
\mbox{sup}_{t \in \mathbb{R}} \left| \widehat F_{T^\ast}(t) -  F_{T}(t)  \right| &\leq& \mbox{sup}_{t \in \mathbb{R}} \left| \widehat F_{T^\ast}(t) - \Phi(t)  \right| + \mbox{sup}_{t \in \mathbb{R}}\left| \Phi(t) - F_{T}(t)  \right|,
\end{eqnarray}
where $\Phi(t)$ is the CDF of a standard normal random variable. 

We show first that  the second term on the right hand side of  equation (\ref{inequality}) converges to zero. Note that
\begin{eqnarray}
 \hat \tau_i^{N_0} = \mu_1(X_i) - \mu_0(X_i) - c + \epsilon_i - \frac{1}{M}  \sum_{j \in \mathcal{J}_M(i)} \epsilon_j  + \mu_0(X_i) - \frac{1}{M}  \sum_{j \in \mathcal{J}_M(i)} \mu_0(X_j).
\end{eqnarray}

Define $\bar \mu_i \equiv  \mu_1(X_i) - \mu_0(X_i) - c$ and $e_i \equiv  \epsilon_i - \frac{1}{M}  \sum_{j \in \mathcal{J}_M(i)} \epsilon_j $. Note that $\{ e_i \}_{i \in \mathcal{I}_1}$ is a triangular array where the distribution of $e_i$ depends on $N_0$ because the nearest neighbors may change when $N_0$ increases. Moreover, $e_i $ may not be independent across $i$ because two treated observations may share the same nearest neighbor. Let $\Omega$ be the event that there is no shared nearest neighbor when there are $N_0$ control observations. We show in Lemma \ref{Lemma_no_shared} that, given the rates in which $N_0$ and $N_1$ diverge, $Pr(\Omega)\rightarrow 1$. Since $\mathbb{E}\left[e_i | \Omega \right] = 0$ and $e_i$ has uniformly bounded fourth moments and second moments uniformly bounded from below,  it follows from the Lindeberg-Feller Theorem that, for all $t$, 
\begin{eqnarray} \label{conv1}
Pr \left(\frac{\sqrt{N_1} \frac{1}{N_1} \sum_{i \in \mathcal{I}_1} e_i }{\sqrt{\frac{1}{N_1}  \sum_{i \in \mathcal{I}_1} var( e_i | \Omega) }} \leq t \middle| \Omega \right)   \rightarrow \Phi(t).
\end{eqnarray}

We also have that, for $i \in \mathcal{I}_1$, $\bar \mu_i $ is iid, $\mathbb{E}[\bar \mu_i ] = 0$ under the null, and $\bar \mu_i $ has finite second moments. Therefore,  $N_1^{-1/2} \sum_{i \in \mathcal{I}_1 }\bar \mu_i \buildrel d \over \rightarrow N(0,var(\bar \mu_i))$. Even though we cannot guarantee that $\mathbb{E} [\bar \mu_i | \Omega]$ equals zero, given that $Pr(\Omega)\rightarrow 1$, it follows that, conditional on $\Omega$, we also have $N_1^{-1/2} \sum_{i \in \mathcal{I}_1} \bar \mu_i  \buildrel d \over \rightarrow N(0,var(\bar \mu_i ))$. Finally, from Corollary 1 from AI, note that $\sqrt{N_1} \frac{1}{N_1} \sum_{i \in \mathcal{I}_1} [ \mu_0(X_i) - \frac{1}{M}  \sum_{j \in \mathcal{J}_M(i)} \mu_0(X_j)] \buildrel p \over \rightarrow 0$. Again, this convergence in probability is also valid if we condition on $\Omega$, because $Pr(\Omega)\rightarrow 1$.  Combining all these results, we have that 
\begin{eqnarray} \label{conv1}
Pr \left(\frac{\sqrt{N_1} \frac{1}{N_1} \sum_{i \in \mathcal{I}_1} \hat \tau_i^{N_0} }{\sqrt{\frac{1}{N_1}  \sum_{i \in \mathcal{I}_1} var( \bar \mu_i + e_i | \Omega) }} \leq t \middle|  \Omega \right) \rightarrow \Phi(t).
\end{eqnarray}
 
Moreover, from WLLN for triangular arrays, conditional on $\Omega$, we have
\begin{eqnarray}\label{conv2}
\frac{1}{N_1} \sum_{i \in \mathcal{I}_1}  \hat \tau_i^{N_0} \buildrel p \over \rightarrow 0, \mbox{ and } \frac{ \frac{1}{N_1} \sum_{i \in \mathcal{I}_1} var( \bar \mu_i + e_i | \Omega) }{ \frac{1}{N_1-1} \sum_{i \in \mathcal{I}_1} ( \hat \tau_i^{N_0})^2 } \buildrel p \over \rightarrow 1.
\end{eqnarray}
 
Combining equations (\ref{conv1}) and (\ref{conv2}), we have that, conditional on $\Omega$, $\sqrt{N_1} \widehat T \buildrel d \over \rightarrow N(0,1)$. Since $Pr(\Omega) \rightarrow 1$, we also have that  $\sqrt{N_1} \widehat T$ converges in distribution to a standard normal, which is  a continuous random variable. Therefore, this implies $\mbox{sup}_{t \in \mathbb{R}}\left|F_{T}(t) - \Phi(t) \right| \rightarrow 0$. 

We consider now the  first term on the right hand side of  equation (\ref{inequality}). Note that
\begin{eqnarray}
\sqrt{N_1} \widehat T^\ast = \frac{ \sum_{i \in \mathcal{I}_1} g_i \hat \tau_i^{N_0} }{\left( \sum_{i \in \mathcal{I}_1} (\hat \tau_i^{N_0})^2\right)^{1/2}} \left(  \frac{\frac{1}{N_1} \sum_{i \in \mathcal{I}_1}( \hat \tau_i^{N_0})^2}{{\frac{1}{N_1 - 1} \sum_{i \in \mathcal{I}_1}  (\hat \tau_i^{N_0} )^2 - \frac{N_1}{N_1 - 1} \left( \frac{1}{N_1} \sum_{i \in \mathcal{I}_1} g_i \hat \tau_i^{N_0} \right)^2  }} \right)^{1/2},
\end{eqnarray}

We need to show that, conditional on the realizations of $\hat \tau_i^{N_0}$, ${ \sum_{i \in \mathcal{I}_1} g_i \hat \tau_i^{N_0} }/{\left( \sum_{i \in \mathcal{I}_1} (\hat \tau_i^{N_0})^2\right)^{1/2}} \buildrel d \over \rightarrow N(0,1)$, and $\frac{1}{N_1} \sum_{i \in \mathcal{I}_1} g_i \hat \tau_i^{N_0} \buildrel p \over \rightarrow 0$. Let $s_i = g_i \hat \tau_i^{N_0}/{\left( \sum_{j \in \mathcal{I}_1} (\hat \tau_j^{N_0})^2\right)^{1/2}}$. Then $\mathbb{E}^\ast[s_i] = 0$, and $\mathbb{E}^\ast[s_i^2] = (\hat \tau_i^{N_0})^2/\left( \sum_{j \in \mathcal{I}_1} (\hat \tau_j^{N_0})^2\right)$, implying that $\sum_{i \in \mathcal{I}_1} \mathbb{E}^\ast[s_i^2]  = 1$. Therefore, we only need that  $\sum_{i \in \mathcal{I}_1} \mathbb{E}^\ast[ | s_i |^{2+\delta}]  \rightarrow 0$ for some $\delta>0$ to guarantee the Lyapunov condition and apply the Lindeberg-Feller CLT.  
\begin{eqnarray}
\sum_{i \in \mathcal{I}_1} \mathbb{E}^\ast[ | s_i |^{2+\delta}] =  \frac{1}{N_1^{\delta/2}} \frac{  \frac{1}{N_1}  \sum_{i \in \mathcal{I}_1} |\hat \tau_i^{N_0}|^{2+\delta}}{\left(  \frac{1}{N_1}  \sum_{j \in \mathcal{I}_1} (\hat \tau_j^{N_0})^2 \right)^{1+\delta/2}}.
\end{eqnarray}

Now consider the distribution of $\hat \tau_i^{N_0}$. Let $S_{N_1} =1$ if there is a shared nearest neighbor, and zero otherwise. Consider a triangular array  $\hat \tau_i^{N_0}$ conditional on a sequence $\{S_{N_1}  \}_{N_1 \in \mathbb{N}}$ such that $S_{N_1} \rightarrow 0$. In this case, for all $N_0$ greater than some $\bar N$,  we have that $\hat \tau_i^{N_0}$ is independent across $i$ and, for some $\gamma>0$, $\mathbb{E}\left[ |\hat \tau_i^{N_0}|^{4+\gamma}  \right]$ is uniformly bounded.\footnote{We use here Assumption \ref{Assumption_appendix2}, and the fact that $\mu_w(x)$ is continuous and the support of $X_i$ is compact. }
 Therefore, for $\delta$ and $\Delta$ sufficiently small, we have that $\mathbb{E}[(|\hat \tau_i^{N_0}]^{2+\delta})^{2+\Delta}]$ is uniformly bounded.  Let $\mathcal{A}$ be the event $\limsup \frac{1}{N_1}  \sum_{i \in \mathcal{I}_1} |\hat \tau_i^{N_0}|^{2+\delta} < c_1$ for some constant $c_1$. Since we can apply a strong law of large numbers for $ \frac{1}{N_1}  \sum_{i \in \mathcal{I}_1} |\hat \tau_i^{N_0}|^{2+\delta} $ conditional on  $\{S_{N_1}  \}_{N_1 \in \mathbb{N}}$, we have that, for some $c_1$, $Pr(\mathcal{A} | S_{N_1} \rightarrow 0) = 1$. Now using Lemma \ref{Lemma_no_shared2}, we have that 
 $Pr(\mathcal{A} ) = Pr(\mathcal{A} | S_{N_1} \rightarrow 0) Pr(S_{N_1} \rightarrow 0) + Pr(\mathcal{A} | S_{N_1}\not \rightarrow 0) Pr(S_{N_1} \not \rightarrow 0) =1$, which implies that $\limsup \frac{1}{N_1}  \sum_{i \in \mathcal{I}_1} |\hat \tau_i^{N_0}|^{2+\delta} < c_1$ with probability one even if we do not condition on $ S_{N_1} \rightarrow 0$. Likewise, since the second moments of $\hat \tau^{N_0}_i$ are uniformly bounded from below, we also have that  $\liminf \frac{1}{N_1}  \sum_{i \in \mathcal{I}_1} |\hat \tau_i^{N_0}|^{2}>c_2$ with probability one for a constant $c_2$, implying that  $\sum_{i \in \mathcal{I}_1} \mathbb{E}^\ast[ | s_i |^{2+\delta}] \rightarrow 0$ with probability one.   Therefore, with probability one, the realizations of  $\hat \tau_i^{N_0}$ are such that, conditional on such realizations,  ${ \sum_{i \in \mathcal{I}_1} g_i \hat \tau_i^{N_0} }/{\left( \sum_{i \in \mathcal{I}_1} (\hat \tau_i^{N_0})^2\right)^{1/2}} \buildrel d \over \rightarrow N(0,1)$.

Moreover, we have that $\mathbb{E}^\ast [ g_i \hat \tau_i^{N_0}] = 0$ and  $var^\ast [ g_i \hat \tau_i^{N_0}] = (\hat \tau_i^{N_0})^2$. If $\mbox{lim}_{N_1 \rightarrow \infty}\frac{1}{N_1} \sum_{i \in \mathcal{I}_1} (\hat \tau_i^{N_0})^2<\infty$, then $\frac{1}{N_1} \sum_{i \in \mathcal{I}_1} g_i \hat \tau_i^{N_0} \buildrel p \over \rightarrow 0$ given the measure induced by the sign changes. Again, since $\frac{1}{N_1} \sum_{i \in \mathcal{I}_1} (\hat \tau_i^{N_0})^2$ converges almost surely to a positive constant, this condition is satisfied with probability one. Combining all these results,  we have that  $\sqrt{N_1} \widehat T^\ast \buildrel d \over \rightarrow N(0,1)$, implying that  $ \mbox{sup}_{t \in \mathbb{R}} \left| \widehat F_{T^\ast}(t) -  \Phi(t)  \right| \buildrel p \over \rightarrow 0$. Therefore, $\mbox{sup}_{t \in \mathbb{R}} \left| \widehat F_{T^\ast}(t) -  F_{T}(t)  \right| \buildrel p \over \rightarrow 0$.
 \end{proof}

\section{Other results}

\subsection{Randomization Inference Test Based on Permutations} \label{Section_permutation}

We  consider  an alternative randomization inference test based on permutations.\footnote{A test based on permutations has been studied in the context of an approximate symmetry assumption  by  \cite{Canay_RDD} for regression discontinuity designs.} We consider again, without loss of generality, that $i=1,...,N_1$ are the treated observations. 

 Consider the following alternative  function of the data 
\begin{eqnarray}
\tilde S_{N_0} = \left( \tilde S^0_{N_0,1}, \tilde S^1_{N_0,1},..., \tilde S^M_{N_0,1},...,\tilde S^0_{N_0,N_1}, \tilde S^1_{N_0,N_1},..., \tilde S^M_{N_0,N_1}  \right)'
\end{eqnarray}
where $\tilde S^0_{N_0,i}=Y_i - c$ and $\tilde S^m_{N_0,i} = Y^i_{(m)}$ for $m=1,...,M$. That is, $\tilde S_{N_0}$ is a vector containing the outcomes of the treated observations (minus $c$) and of their $M$-nearest neighbors. The distribution of  $\tilde S_{N_0}$ depends on $N_0$, because the quality of the matches will depend on $N_0$. In this notation, and considering $c=0$, the matching estimator is given by
 \begin{eqnarray}
 \hat \tau = \frac{1}{N_1} \sum_{i=1}^{N_1}\left(  \tilde S^0_{N_0,i} - \frac{1}{M} \sum_{j=1}^M  \tilde S^j_{N_0,i} \right).
 \end{eqnarray}

Let $\tilde G_i$ be the set of all permutations  $\pi_i = (\pi_i(0),...,\pi_i(M))$ of $\{0,1,...,M  \}$, $ \pi = \otimes_{i=1}^{N_1} \pi_i$, and $\tilde{\textbf{G}} = \otimes_{i=1}^{N_1} \tilde G_i$.  Note that $\tilde{\textbf{G}} $ is the set of all permutations that reassign the treatment status conditional on having exactly one treated observation for each group of treated observation $i$ and its $M$ nearest neighbors. For a given $\pi \in \tilde{\textbf{G}}  $, consider $ \tilde S^\pi_{N_0} = \left( \tilde S^{\pi_1(0)}_{N_0,1}, \tilde S^{\pi_1(1)}_{N_0,1},..., \tilde S^{\pi_1(M)}_{N_0,1},...,\tilde S^{\pi_{N_1}(0)}_{N_0,N_1}, \tilde S^{\pi_{N_1}(1)}_{N_0,N_1},..., \tilde S^{\pi_{N_1}(M)}_{N_0,N_1}  \right)'$.

Let $\tilde K = |\tilde{\textbf{G}}|$ and denote by
\begin{eqnarray}
\tilde T^{(1)}(\tilde S_{N_0}) \leq \tilde T^{(2)}(\tilde S_{N_0}) \leq ... \leq \tilde T^{(\tilde K)}(\tilde S_{N_0})
\end{eqnarray}
the ordered values of $\{ \tilde T( \tilde S^\pi_{N_0}) : \pi \in \tilde{\textbf{G}} \}$, where
 \begin{eqnarray} \label{test_stat}
 \tilde T( \tilde S^\pi_{N_0}) = \left| \frac{1}{N_1} \sum_{i=1}^{N_1}\left(  \tilde S^{\pi_i(0)}_{N_0,i} - \frac{1}{M} \sum_{j=1}^M  \tilde S^{\pi_i(j)}_{N_0,i} \right) \right|.
 \end{eqnarray}
 
We set $\tilde k=\lceil \tilde K(1-\alpha) \rceil$, where $\alpha$ is the significance level of the test, and define the decision rule of the test as
\begin{eqnarray}  \label{decision_permutation}
\tilde \phi(S_{N_0}) = \begin{cases} 1 & \mbox{ if } \tilde T(\tilde S_{N_1}) > \tilde T^{(\tilde k)}(\tilde S_{N_1})      \\ 0 & \mbox{ if } \tilde T(\tilde S_{N_1}) \leq \tilde T^{(\tilde k)}(\tilde S_{N_1}).   \end{cases}
\end{eqnarray}

In words, we calculate the test statistic $\tilde T(\tilde S^\pi_{N_0} )$ for all possible permutations in $\tilde{\textbf{G}}$, and then we reject the null if the actual test statistic  $\tilde T(\tilde S_{N_0} )$ is large relative to the distribution given by these permutations. If $N_1 >1$, we could also consider a standardized test statistic

 \begin{eqnarray} \label{test_std}
 \tilde T^{\mbox{\tiny std}}( \tilde S^\pi_{N_0}) = \frac{\left| \frac{1}{N_1} \sum_{i=1}^{N_1}\left(  \tilde S^{\pi_i(0)}_{N_0,i} - \frac{1}{M} \sum_{j=1}^M  \tilde S^{\pi_i(j)}_{N_0,i} \right)  \right|}{ \sqrt{ \frac{1}{N_1-1} \sum_{i=1}^{N_1}\left(  \tilde S^{\pi_i(0)}_{N_0,i} - \frac{1}{M} \sum_{j=1}^M  \tilde S^{\pi_i(j)}_{N_0,i} - \tilde \tau^\pi \right)^2}},
 \end{eqnarray}
where  $ \tilde \tau^\pi = \frac{1}{N_1} \sum_{i=1}^{N_1}\left(  \tilde S^{\pi_i(0)}_{N_0,i} - \frac{1}{M} \sum_{j=1}^M  \tilde S^{\pi_i(j)}_{N_0,i} \right)$, and consider a decision rule as in \ref{decision_permutation}.

As common in permutation tests, we need a homoskedasticity/treatment effect homogeneity assumption, which we define below.

\begin{assumption} \label{Assumption_iid}
(homoskedasticity/treatment effect homogeneity)  
$(Y(1)  - \tau ) | (X=x,W=1)$ has the same distribution as $Y(0) | (X=x,W=0)$ for all $x \in \mathbb{X}_1$.

\end{assumption}

This is a strong assumption, and excludes, for example, the possibility that treatment affects that variance of the outcome. This assumption also excludes the possibility that treatment has heterogeneous effects depending on $X_i$.

\begin{proposition} \label{test_permutation}

Suppose Assumptions \ref{sample}, \ref{CIA},  \ref{overlap},  \ref{assumption_mu}(b), and  \ref{Assumption_iid} hold. Assume also that  the distribution of $Y$ is  continuous. If we consider the problem of testing $H_0:  \tau =c$, then, for any $\alpha \in (0,1)$, $\mbox{limsup}_{N_0 \rightarrow \infty} \mathbb{E} \left[ \tilde \phi(\tilde S_{N_0}) \right] \leq \alpha$ when $N_1$ is fixed. 

\end{proposition}

\begin{proof}
Again, we apply Theorem 3.1 from  \cite{Canay}. First, note that, for $i \in \mathcal{I}_1$,  $Y_i$ is distributed as  $Y(1)|W=1$.  Therefore,  under Assumption \ref{Assumption_iid}, we have that, under the null,  $\tilde S^0_{N_0,i} = (Y_i - c)$ is distributed as $Y(0)|W=1$. 

Now consider  $\tilde S^m_{N_0,i} = Y^i_{(M)}$. If we condition on $X_i = \bar x$, we know from Lemma \ref{Lemma_convergenceY} that $ Y^i_{(M)}$ converges in distribution to $Y(0)|X = \bar x$. Now, to derive the unconditional asymptotic distribution, let $h(y)$ be a continuous and bounded function. Then 
\begin{eqnarray} 
\mathbb{E}[h(Y^i_{(M)}) ]&=&\mathbb{E} \left\{\mathbb{E}[h(Y^i_{(M)}) | X_i] \right\}  \rightarrow \mathbb{E} \left\{\mathbb{E}[h(Y_i(0)) | X_i]  \right\} \\
&=& \mathbb{E}[h(Y_i(0) ) ] = \mathbb{E}[h(Y(0) ) \mid W=1 ].
\end{eqnarray}

Therefore,   $\tilde S^m_{N_0,i} = Y^i_{(M)}$ also converges in distribution to $Y(0) | W=1$. Since, conditional on $X_i= \bar x$, the limiting distribution of $\tilde S^m_{N_0,i}$ is independent across $m$ for all $m=0,1,...,M$, we have that the limiting distribution of $(\tilde S^0_{N_0,i},...,\tilde S^M_{N_0,i})$ is invariant to permutations. 
Moreover, from Lemma \ref{Lemma_no_shared}, the vectors $(\tilde S^0_{N_0,i},...,\tilde S^M_{N_0,i})$ are mutually independent across $i$. Combining these results, the asymptotic distribution of $\tilde S_{N_0} $ is invariant to  transformations in $\tilde{\textbf{G}}$.  

We also have that the test statistic function $\tilde T(\tilde S)$ is continuous. Finally, we show that, for two distinct elements $\pi \in \tilde G$ and $\pi' \in \tilde G$, either $\tilde T(\tilde S^\pi ) = \tilde T(\tilde S^{\pi'})$ for all possible realizations of $\tilde S$, or $Pr(\tilde T(\tilde S^\pi ) \neq \tilde T(\tilde S^{\pi'})) = 1$. Suppose $M>1$. Then, if $\pi$ and $\pi'$ are such that $\pi_i(0) = \pi'_i(0)$ for all $i \in \mathcal{I}_1$, then we will have  $\tilde T(\tilde S^\pi ) = \tilde T(\tilde S^{\pi'})$ for all possible realizations of $\tilde S$. If $\pi$ and $\pi'$ are such that $\pi_i(0) \neq \pi'_i(0)$ for at least one $i \in \mathcal{I}_1$, then the probability that  $\tilde T(\tilde S^\pi ) = \tilde T(\tilde S^{\pi'})$  would be equal to zero, because $\tilde S$ is a continuous random variable. For the case $M=1$, we would have  $\tilde T(\tilde S^\pi ) = \tilde T(\tilde S^{\pi'})$ for all possible realizations of $\tilde S$ if $\pi$ and $\pi'$ are such that $\pi_i(0) = \pi'_i(0)$ for all $i$, or $\pi_i(0) = \pi'_i(1)$ for all $i$. Otherwise, $Pr(\tilde T(\tilde S^\pi ) \neq \tilde T(\tilde S^{\pi'})) = 1$.  
\end{proof}

Again, the main intuition of the proof is that the limiting distribution of $\tilde S_{N_0}$, under the null, is invariant to the transformations in $\tilde{\textbf{G}}$.  Assumption \ref{Assumption_iid} is crucial for this result, when we consider a setting with fixed $N_1$. To understand that, suppose, for example, that the null is true, but  $\mathbb{V}[Y(1)|X=x,W=1] > \mathbb{V}[Y(0)|X=x,W=0]$. If $M>1$, then a permutation that uses control observations in place of  treated ones would have a less volatile distribution relative to the distribution of the matching estimator. This would lead to a rejection rate higher than $\alpha$.\footnote{The intuition is similar to the one presented by \cite{FP_2018} for the DID estimator. The matching estimator will compare the averages of $N_1$ treated observations with the average of $M \times N_1$ control nearest neighbors. Therefore, if $M>1$, then the variance of the treated observations will have a relatively larger impact on the variance of the matching estimator than the variance of the control observations. As a consequence, permutations that place control observations as treated would have a lower variance than the actual estimator if $\mathbb{V}[Y(1)|X=x,W=1] < \mathbb{V}[Y(0)|X=x,W=0]$. Following the same logic, this also implies that such a test may have a low power if the treatment decreases the variance of the outcome (that is, $\mathbb{V}[Y(1)|X=x,W=1] < \mathbb{V}[Y(0)|X=x,W=0]$). } This is an important drawback of this test, because we may reject at a rate higher than $\alpha$ even when the null is true if we have such heteroskedasticity. This is a cost we have to bear in order to have an inference method that is valid when we consider a setting with $N_1$ is fixed. 

This permutation test is similar in spirit to the test proposed by \cite{CT} for differences-in-differences with few treated and many control groups. Note that they also need a homoskedasticity assumption, highlighting the fact that we need to rely on stronger assumptions if we want to construct a test that is valid regardless of the number of treated observations. \cite{FP_2018}  relax in some sense such assumption considered by  \cite{CT}, allowing for specific forms of heteroskedasticity based on the observed covariates. Interestingly, since Assumption \ref{Assumption_iid} only requires  homoskedasticity conditional on $X_i$, what we are doing is similar in spirit to the idea of a non-parametric version of \cite{FP_2018}. Therefore, while Assumption \ref{Assumption_iid} may still be restrictive in some settings, it is a relatively weaker assumption than that considered by  \cite{CT}.

Proposition \ref{test_permutation} remains valid if we consider $\tilde T^{\mbox{\tiny std}}( \tilde S^\pi_{N_0})$ instead of $ \tilde T( \tilde S^\pi_{N_0})$ as test statistic. If we consider this alternative test statistic, then we show that we can relax Assumption \ref{Assumption_iid} when we consider a setting in which $N_1 \rightarrow \infty$ at a lower rate relative to $N_0$.  

\begin{proposition} \label{permutation_appendix}

Suppose Assumptions \ref{sample}, \ref{CIA}, \ref{overlap},   \ref{Assumption_rates}, and \ref{Assumption_appendix2} hold. If we consider the problem of testing $H_0: \tau=c$, then the permutation test using $\tilde T^{\mbox{\tiny std}}( \tilde S_{N_0}) $ as test statistic  is asymptotically valid when  $N_1,N_0 \rightarrow \infty$. 

\end{proposition}

\begin{proof}
Note that $\tilde T^{\mbox{\tiny std}}( \tilde S_{N_0}) = T(S_{N_0})$. That is, the test statistic used in the sign-changes test is the same as the standardized test statistic considered in the permutation test when we consider the original permutation. Therefore, we know from the proof of Proposition \ref{sign_changes_asympt} that 
\begin{eqnarray}
\sqrt{N_1} \widetilde T =  \frac{ \frac{1}{\sqrt{N_1}} \sum_{i=1}^{N_1}\left(  \tilde S^{0}_{N_0,i} - \frac{1}{M} \sum_{j=1}^M  \tilde S^{j}_{N_0,i} \right)  }{ \sqrt{ \frac{1}{N_1-1} \sum_{i=1}^{N_1}\left(  \tilde S^{0}_{N_0,i} - \frac{1}{M} \sum_{j=1}^M  \tilde S^{j}_{N_0,i} - \tilde \tau \right)^2}} \buildrel d \over \rightarrow N(0,1)
\end{eqnarray}

We consider now the distribution of the permutations 
\begin{eqnarray}
\sqrt{N_1} \widetilde T^\ast =  \frac{ \frac{1}{\sqrt{N_1}} \sum_{i=1}^{N_1}\left(  \tilde S^{\pi_i(0)}_{N_0,i} - \frac{1}{M} \sum_{j=1}^M  \tilde S^{\pi_i(j)}_{N_0,i} \right)  }{ \sqrt{ \frac{1}{N_1-1} \sum_{i=1}^{N_1}\left(  \tilde S^{\pi_i(0)}_{N_0,i} - \frac{1}{M} \sum_{j=1}^M  \tilde S^{\pi_i(j)}_{N_0,i} - \tilde \tau^\pi \right)^2}}.
\end{eqnarray}

Let $\widetilde{Pr}$ denote the probability measure induced by the permutations, and $\widetilde{\mathbb{E}}$ the corresponding expectation conditional on a realization of the random variables $X_i$ and $\epsilon_i$. Define  $\widetilde F_{T^\ast}(t) \equiv \widetilde{Pr} \left( \sqrt{N_1} \widetilde T^\ast  \leq t \right) $. Note that, given the realizations of $X_i$ and $\epsilon_i$, we can derive $\widetilde F_{T^\ast}(t) $ by considering all permutations $\pi \in \mathbf{ \tilde G}$ with equal probabilities, or approximate that as well as we want by drawing bootstrap samples of $\pi \in \mathbf{ \tilde G}$.  Given that $\sqrt{N_1} \widetilde T \buildrel p \over \rightarrow N(0,1)$, we only need to show that  $\mbox{sup}_{t \in \mathbb{R}} \left| \widetilde F_{T^\ast}(t) -  \Phi(t)  \right| \buildrel p \over \rightarrow 0$.

Similar to the proof of Proposition \ref{sign_changes_asympt}, note that 
\begin{eqnarray}
\sqrt{N_1} \widetilde T^\ast &=&  \frac{ \frac{1}{\sqrt{N_1}} \sum_{i=1}^{N_1}\left(  \tilde S^{\pi_i(0)}_{N_0,i} - \frac{1}{M} \sum_{j=1}^M  \tilde S^{\pi_i(j)}_{N_0,i} \right)}{\left(\frac{1}{N_1} \sum_{i \in \mathcal{I}_1} \left[ \sum_{j=0}^M \frac{1}{M+1}\left(  \tilde S^{j}_{N_0,i} - \frac{1}{M} \sum_{j' \neq j}  \tilde S^{j'}_{N_0,i} \right)^2 \right] \right)^{1/2}   } \times \\
&& \times    \left(\frac{\frac{1}{N_1} \sum_{i \in \mathcal{I}_1} \left[ \sum_{j=0}^M \frac{1}{M+1}\left(  \tilde S^{j}_{N_0,i} - \frac{1}{M} \sum_{j' \neq j}  \tilde S^{j'}_{N_0,i} \right)^2 \right]    }{ { \frac{1}{N_1-1} \sum_{i=1}^{N_1}\left(  \tilde S^{\pi_i(0)}_{N_0,i} - \frac{1}{M} \sum_{j=1}^M  \tilde S^{\pi_i(j)}_{N_0,i} - \tilde \tau^\pi \right)^2}}\right)^{1/2}.
\end{eqnarray}

Note that 
\begin{eqnarray} \nonumber
\widetilde{\mathbb{E}} \left[ \tilde S^{\pi_i(0)}_{N_0,i} - \frac{1}{M} \sum_{j=1}^M  \tilde S^{\pi_i(j)}_{N_0,i} \right] = 0\mbox{, and } \widetilde{\mathbb{E}} \left[ \left(\tilde S^{\pi_i(0)}_{N_0,i} - \frac{1}{M} \sum_{j=1}^M  \tilde S^{\pi_i(j)}_{N_0,i} \right)^2 \right] = \sum_{j=0}^M \frac{1}{M+1}\left(  \tilde S^{j}_{N_0,i} - \frac{1}{M} \sum_{j' \neq j}  \tilde S^{j'}_{N_0,i} \right)^2 .
\end{eqnarray}

Therefore, following similar steps as in the proof of Proposition \ref{sign_changes_asympt}, we have that, with probability one, the realizations of $ \tilde S^{j}_{N_0,i} $ are such that $\widetilde{Pr} \left(\sqrt{N_1} \widetilde T^\ast \leq t \right) \buildrel p \over \rightarrow \Phi(t)$, which completes the proof.
\end{proof}

Therefore, the test is asymptotically valid under stronger assumption for any any fixed $N_1$, and valid under weaker assumptions when $N_1$  diverges.  Note that a permutation tests based on  $ \tilde T( \tilde S^\pi_{N_0})$ remains invalid if Assumption \ref{Assumption_iid} does not hold even when $N_1 \rightarrow \infty$, highlighting the importance of the choice of the test statistic.

\begin{remark}
\normalfont
Remark \ref{ties} also applies to this test.
\end{remark}

\begin{remark}
\normalfont  \label{finite_sample_adjustment}

Similar to Remark \ref{finite_sample_adjustment2}, we propose a finite sample fix in the permutation test. If a control observation is the nearest neighbor for two or more treated observations, then we  restrict to permutations of $ S_{N_0}$ such that this control observation is always placed as either treated or control.  Since the probability that two treated observations share the same nearest neighbor goes to zero when $N_1$ is fixed and $N_0 \rightarrow \infty$, for a fixed $M$, this finite sample adjustment is asymptotically irrelevant. Such adjustment prevents over-rejection with finite $N_0$ in cases different treated observations share the same nearest neighbor.

\end{remark}

\begin{remark}
\normalfont 
This test is also asymptotically valid for bias-corrected matching estimators, as presented in equation (\ref{biasadj}). In this case, we define  $\tilde S^0_{N_0,i}=Y_i-c$ and $\tilde S^m_{N_0,i} = Y^i_{(m)} + \hat \mu_0(X_i) - \hat \mu_0(X^i_{(m)})$. The key idea is that, again, $\tilde S_{N_0,i}^m = Y^i_{(m)} + \hat \mu_0(X_i) - \hat \mu_0(X^i_{(m)}) \buildrel d \over \rightarrow Y_i(0)|X_i$, for all $m=1,...,M$, because $\hat \mu_0(X_i) - \hat \mu_0(X^i_{(m)}) \buildrel p \over \rightarrow 0$.
\end{remark}

\begin{remark}
\normalfont
When $M=1$, the test based on sign changes  and the one based on  permutations are numerically equivalent.

\end{remark}

\subsection{Conditional Average Treatment Effects on the Treated (CATT)} \label{Appendix_CATT}

An alternative estimand we can consider in this setting is the average treatment effect on the treated, conditional on the realization of the covariates of the treated. Without loss of generality, let $i=1,...,N_1$ be the treated observations, and let  $\mathbf{X} = [X_1 ~ X_2 ~ \cdots ~ X_{N_1}]$. Fix $ \bar{\mathbf{x}} =  [\bar x_1 ~ \bar x_2 ~ \cdots ~ \bar x_{N_1}] \in \otimes_{i=1}^{N_1} \mathbb{X}_1 $. Then, we define the conditional average treatment effect on the treated (CATT) by
\begin{eqnarray}
\tau(\bar{\mathbf{x}}) \equiv \frac{1}{N_1} \sum_{i \in \mathcal{I}_1}  \mathbb{E} \left[ Y_i(1) - Y_i(0) | X_i = \bar x_i  \right] =  \frac{1}{N_1}  \sum_{i \in \mathcal{I}_1} \left[  \mu_1(\bar x_i)  - \mu_0 ( \bar x_i)   \right].
\end{eqnarray}

In this case, we can show that the NN matching estimator is asymptotically unbiased for this estimand. We consider a fixed $\bar{\mathbf{x}}$ such that $\bar x_i \neq \bar x_j$ for all $i \neq j$. Since the distribution of $X_i$ is continuous, the realization $\bar{\mathbf{x}}$ satisfies this condition with probability one.

\begin{proposition} 
\label{Prop_conditional}
Let $\bar{\mathbf{x}} \in \otimes_{i=1}^{N_1} \mathbb{X}_1$ such that $\bar x_i \neq \bar x_j$ for all $i \neq j$. 

(1) Under Assumptions \ref{sample}, \ref{CIA},  \ref{overlap}, and \ref{assumption_mu}(a), $\mathbb{E}[\hat \tau | \mathbf{X} = \bar{\mathbf{x}}] \rightarrow \tau( \bar{\mathbf{x}}) $ when $N_0 \rightarrow \infty$ and $N_1$ is fixed.

(2) Under Assumptions \ref{sample}, \ref{CIA},  \ref{overlap}, and \ref{assumption_mu}(b),  conditional on $ \mathbf{X} = \bar{\mathbf{x}}$,
\begin{eqnarray} \nonumber
\hat \tau \buildrel d \over \rightarrow  \tau(\bar{\mathbf{x}}) + \frac{1}{N_1}  \sum_{i \in \mathcal{I}_1} \left( \epsilon_i - \frac{1}{M}  \sum_{m=1}^M \eta_i^m \right)  \mbox{ when $N_0 \rightarrow \infty$ and $N_1$ is fixed}, 
\end{eqnarray}
where $\eta_i^m  \buildrel d \over = (Y(0) - \mu_0(\bar x))|(X = \bar x_i, W=0) $ for $i \in \mathcal{I}_1$, and $\eta_i^m$ is independent across  $m$ and $i$, and independent of all $\epsilon_j$.

\end{proposition}

\begin{proof}
The first part of this proposition was already proved as an intermediate step in the Proof of Proposition \ref{unbiased} (equation \ref{eq_unbiased}). The second part of this proposition also follows from Lemma \ref{Lemma_convergenceY}. Now we just have to show that $\theta_i^m$ is independent across $m$ and $i$. Since $\bar{\mathbf{x}}$ is such that $\bar x_i \neq \bar x_j$ for all $i \neq j$,  it must be that there is a $c>0$ such that $d(\bar x_i,\bar x_j)>c$ for all $i,j \in \mathcal{I}_1$ with $i \neq j$. However, we know that $Pr(d(\bar x_i,X^i_{(m)})>e) \rightarrow 0$ for all $e >0$. Therefore, the probability that $k \in \mathcal{I}_0$ belongs to $\mathcal{J}_M(i)$ and $\mathcal{J}_M(j)$ converges to zero.  
\end{proof}

We can also consider the conditions in which the sign-changes test is valid to test hypotheses of the type $\tau(\bar{\mathbf{x}}) = c$. In this case, instead of Assumption \ref{Assumption_symmetry}, we have to consider a strong assumption. 

\begin{assumption} \label{Assumption_cond}
Let $\xi_x \sim (Y(1) - \mu_1(x)) | (X=x,W=1)$ and  $\eta^m_x  \sim (Y(0) - \mu_0(x)) | (X=x,W=0)$, and consider $(\xi_x,\eta^1_x, \cdots,\eta^M_x)$ mutually independent. Then  $\xi_x - \frac{1}{M} \sum_{m=1}^M \eta^m_x$  is symmetric around zero and $\mu_1(x) - \mu_0(x) = \tau$ for all $x \in \mathbb{X}_1$.

\end{assumption}

Note that this assumption implies that  treatment effects are homogeneous (instead of assuming that the distribution of treatment effects is symmetric). It also generally  implies that  the distribution of potential outcomes, conditional on $X=x$, are symmetric around their conditional mean. If $M=1$, then this assumption is also satisfied if $\xi_x $ and $ \eta^m_x$ have the same distribution, even if this distribution is not symmetric.

\begin{proposition} 
Let $\bar{\mathbf{x}} \in \otimes_{i=1}^{N_1} \mathbb{X}_1$ such that $\bar x_i \neq \bar x_j$ for all $i \neq j$, and suppose Assumptions \ref{sample}, \ref{CIA},  \ref{overlap},  \ref{assumption_mu}(b), and  \ref{Assumption_cond} hold. Assume also that the distribution of $Y$ is  continuous. If we consider the problem of testing $H_0:  \tau(\bar{\mathbf{x}}) =c$, then, for any $\alpha \in (0,1)$, $\mbox{limsup}_{N_0 \rightarrow \infty} \mathbb{E} \left[  \phi( S_{N_0}) | \mathbf{X} = \bar{\mathbf{x}} \right]  \leq \alpha$ when $N_1$ is fixed.

\end{proposition}

\begin{proof}
We apply again Theorem 3.1 from  \cite{Canay}. We first show that, when $N_0 \rightarrow \infty$, the limiting distribution of $S_{N_0}$, conditional on $\mathbf{X} = \bar{\mathbf{x}}$, is invariant to sign changes  under the null. Let $S$ be this limiting distribution.  This is true if, asymptotically, $\hat \tau_i$ and $\hat \tau_j$ are independent for $i \neq j$, and the distribution of  $\hat \tau_i$ is symmetric around zero. It is not necessary for   $\hat \tau_i$ to have the same distribution across $i$.
Under the null, the asymptotic distribution of $\hat \tau_i^{N_0} $  conditional on $X_i = \bar x_i$  is given by  $ \xi_{\bar x_i} - \frac{1}{M} \sum_{m=1}^M \eta^m_{\bar x_i}$, which is symmetric around zero given Assumption \ref{Assumption_cond}. Moreover, Proposition \ref{Prop_conditional} also shows that, asymptotically, $\hat \tau_i^{N_0}$ are independent across $i$. 

We also have that the test statistic function $ T(S)$ is continuous. Finally, we show that, for two distinct elements $g \in  G$ and $g' \in  G$, either $T( gS ) = T( g'S)$ for all possible realizations of $S$, or $Pr(T( gS) \neq T(g' S)) = 1$.
If $g$ and $g'$ are such that $g_i = g_i'$ for all $i$, or $g_i = - g_i'$ for all $i$, then $T(g S ) = T(g' S)$ for all possible realizations of $S$. Otherwise, given that $S$ is a continuous random variable, $Pr(T(g S ) \neq T(g' S)) = 1$. 
 \end{proof}

\subsection{Particular case: $Y(0)|X$ is normally distributed} \label{condition_normal}

Let  $Y \sim N(\theta, \sigma^2)$. We first want to show that $\tilde h(\theta,\sigma) = \mathbb{E} [h(Y) | \theta,\sigma]$ is continuous and bounded for any $h()$ continuous and bounded. In this case, 
\begin{eqnarray}
\tilde h(\theta,\sigma) =  \int  h(y) \frac{1}{\sqrt{2 \pi}} \frac{1}{\sigma} e^{-\frac{1}{2} \left( \frac{y - \theta}{\sigma} \right)^2} dy.
\end{eqnarray}

Let $g(y,\theta,\sigma) =  h(y) \frac{1}{\sqrt{2 \pi}} \frac{1}{\sigma} e^{-\frac{1}{2} \left( \frac{y - \theta}{\sigma} \right)^2} $. Since $h(y)$ is continuous and bounded, $g(y,\theta,\sigma)$ is integrable for all $(\theta,\sigma)$, and,  for all $y \in \mathbb{R}$, $g(y,\theta,\sigma)$ is continuous in $(\theta,\sigma)$. We now show that there is a neighborhood of $(\theta,\sigma)$ and an integrable function $q: \mathbb{R} \rightarrow \mathbb{R_+}$ such that, for all $(\theta,\sigma)$ in this neighborhood, $|g(y,\theta,\sigma)| \leq q(y)$. 

Consider the neighborhood of $(\theta,\sigma)$ given by  $(\theta-\delta,\theta+\delta) \times (\sigma-\delta,\sigma+\delta)$ (where $\delta$ is sufficiently small so that $\sigma-\delta>0$), and define
\begin{eqnarray}
q(y) = \begin{cases}  |h(y)|  \frac{1}{\sqrt{2 \pi}} \frac{1}{\sigma - \delta} e^{-\frac{1}{2} \left( \frac{y - (\theta+\delta)}{\sigma + \delta} \right)^2}  \mbox{, if } y> \theta + \delta \\  
 |h(y)|  \frac{1}{\sqrt{2 \pi}} \frac{1}{\sigma - \delta}   \mbox{, if } y \in [ \theta - \delta, \theta+\delta] \\
  |h(y)|  \frac{1}{\sqrt{2 \pi}} \frac{1}{\sigma - \delta} e^{-\frac{1}{2} \left( \frac{y - (\theta-\delta)}{\sigma + \delta} \right)^2}  \mbox{, if } y< \theta - \delta
 \end{cases}
\end{eqnarray}

For any $(\theta,\sigma) \in (\theta-\delta,\theta+\delta) \times (\sigma-\delta,\sigma+\delta)$, $ | g(y,\theta,\sigma) | \leq q(y)$,  and $q(y)$ is integrable. Therefore,  $h(\theta,\sigma) $ is continuous at any point $(\theta,\sigma)$. Moreover, since $h(y)$ is bounded,  $\tilde h(\theta,\sigma) $ is also bounded.

Now let $Y(0)|(X=x) \sim N(\theta(x) , \sigma(x))$, where $\theta(x)$ and $\sigma(x)$ are continuous functions. Since compositions of continuous functions are continuous, it follows that $\tilde h(x) =  \int  h(y) \frac{1}{\sqrt{2 \pi}} \frac{1}{\sigma(x)} e^{-\frac{1}{2} \left( \frac{y - \theta(x)}{\sigma(x)} \right)^2}dy $ is  bounded and continuous in $x$.

\subsection{Settings in which not all covariates are continuous} \label{Appendix_discrete}

The case in which we have some covariates that are discrete, but at least one covariate is continuous, can be dealt with by estimating  treatment effects within subsamples defined by the values of the discrete covariates, and then aggregating on such covariates, as argued by AI. The NN matching estimators within these covariate cells --- where the nearest neighbor is chosen based on the continuous covariates, conditional on having the same value of the discrete covariates --- satisfy all conditions for Propositions \ref{unbiased} and \ref{test}.  Note that, since we consider a setting with $N_1$ fixed and $N_0 \rightarrow \infty$, the probability that each treated observations has at least $M$ control observations with the same value for the discrete covariates converges to one.

We also consider the case in which we have variables with mixed distributions. Consider a setting with a single covariate that is equal to $0$ with positive probability, and that  has a continuous distribution in the rest of its support. In this case, we may have some treated observations with $X_i=0$ and some with $X_i \neq 0$. For those with $X_i=0$, we can use all control observations with $X_j=0$ as their matches. In this case, there would not be finite sample bias conditional on having $X_i=0$. For the treated observations with $X_i \neq 0$, we can look for the nearest neighbors as usual, so Proposition \ref{unbiased} would apply.\footnote{Note that, for any $X_i = \bar x \neq 0$, the probability that the nearest neighbor is a control observation with $X_j=0$ converges in probability to zero when $N_0 \rightarrow \infty$. } For the sign-changes test, we can consider the following strategy. We can partition the control observations with $X_j=0$ in $n$ mutually exclusive groups, where $n$ is the number of treated observations with $X_i=0$, and assign one of these partitions as the neighbors to each treated observation. Then we can proceed with  the sign-changes test as usual. Since there would not be overlap in the set of neighbors for different treated observations with $X_i=0$, the test would remain valid.

\section{``Jovem de Futuro Program'' \& Empirical MC} \label{JF}

\subsection{Description of the ``Jovem de Futuro Program'' }

We explore the validity of matching estimators and of different inferential methods in the estimation of the effects of an educational program in Brazil called ``Jovem de Futuro''. This application  provides an example of possible a setting with few treated and many control schools. In Section \ref{MC}, we  conduct an empirical Monte Carlo (MC) study based on this application (e.g. \cite{Huber}), while in Section \ref{empirical_application} we estimate the effects of the program using matching estimators. Here we provide more details on how the empirical MC study was designed, and on the empirical application.

Before we proceed,  we start with a brief description of the program, and we present some descriptive statistics (see \cite{Barros} for more details). The ``Jovem de Futuro'' program, an initiative of the ``Instituto Unibanco'' (Unibanco Institute), aims to improve the quality of education in Brazilian public schools. This is a three-year-long intervention based on two efforts: (i)  providing school managers with strategies and instruments to become more efficient and productive, and (ii) providing conditional cash transfers to schools.\footnote{The conditions are to improve students' performance on a standardized examination by the Institute at the end of each school year and to implement a participatory budget process in the school.} In 2007, the Unibanco Institute created and implemented the program in three schools in Sao Paulo. Then they implemented a few randomized control trials in the following years to evaluate the impact of the program. 

We focus on the 2010 implementation of the program, which took place in Rio de Janeiro and Sao Paulo. Schools in these two states were invited to participate in the program, knowing in advance that they would be randomly assigned to receive the program starting in 2010, or that they would be placed first as a control group and would start the program only in 2013. We use information from the 2007 to 2012 ``Exame Nacional do Ensino M\'{e}dio'' (ENEM), a national exam that evaluates high school students in Brazil, as a measure of students' proficiency.\footnote{It is not possible to identify the schools that participated in the ``Jovem de Futuro'' experiment using the public-access ENEM microdata before 2007. For this reason, we do not consider earlier implementations of the program in Minas Gerais and Rio Grande do Sul, because we would only have one year of pre-treatment outcome.}$^,$\footnote{For 2007 and 2008, we focus on the score on a 63-question multiple-choice test on various subjects (Portuguese, History, Geography, Math, Physics, Chemistry and Biology). Since 2009, the exam has been composed of 180 multiple-choice questions, equally divided into four areas of knowledge: languages, codes and related technologies; human sciences and related technologies; natural sciences and related technologies; and mathematics and its technologies.  In this case, we consider the average score for these four areas. For each year and for each state, we standardize the test scores based on the sample of students from the experimental control schools.  }  Focusing on  schools with test score information from 2007 to 2012, we have 15 treated schools in Rio de Janeiro and 39  in Sao Paulo, with the same number of control schools in each state.\footnote{We exclude one control and two treated schools from Sao Paulo because they lack information for at least one of these years.} 

Column 1 of Table \ref{summary_stat} presents the difference in test scores for treated and control experimental schools in Rio de Janeiro, and  column 3 shows the same difference for schools in Sao Paulo. Panel A presents this information for 2007 to 2009, which was before the intervention. For Rio de Janeiro, all differences are small and not statistically different from zero, as one would expect given random assignment. For Sao Paulo, however, there are significant differences in test scores in 2007 and 2008, suggesting that there may have been some problems in the assignment of treatment schools. Panel B presents the results for the three years after the implementation of the program. The comparison between treated and control schools suggest a null effect of the program in Rio de Janeiro, and a positive and significant effect in Sao Paulo. We should be careful in interpreting the results for Sao Paulo, however, due to the imbalances in pre-intervention test scores.\footnote{\cite{Leonardo_rosa} analyzes the   ``Jovem de Futuro'' program using a differences-in-differences approach, exploiting the experimental design of the program. He finds a positive and significant effect of the program for both Rio de Janeiro and Sao Paulo. There are a few differences in our analyses that justify the different results. First, we consider an intention to treat effect, including schools that abandoned the program after its implementation, while \cite{Leonardo_rosa} includes only strata with no attritors (see \cite{Ferman_attrition} for a discussion on potential bias from the exclusion of strata with attrition problems). Second, \cite{Leonardo_rosa}   considers an exam that was administered on the treated and control schools to evaluate this program. We are not able to use this dataset because this information is not available for  non-experimental schools. Finally, we aggregate our data at the school level, while  \cite{Leonardo_rosa} uses individual-level data.     }

Columns 2 and 4 of Table \ref{summary_stat} present differences in test scores for public  schools that did not participate in the experiment and schools in the experimental control group. In Rio de Janeiro, schools that (voluntarily) decided to participate in the experiment had better outcomes prior to the intervention, relative to other schools that did not participate in the experiment. In Sao Paulo, schools in the experimental control group were, on average, worse than the schools that did not participate in the experiment. Interestingly, Rio de Janeiro has 966 and Sao Paulo has 3481 non-experimental public schools, thus providing a setting with few treated  and many (non-experimental) control schools.

\subsection{Empirical Monte Carlo Study}  \label{Appendix_MC}

We consider an empirical MC study based on the ``Jovem de Futuro'' implementation.  We first estimate a probit model using schools' average test scores in the three years prior to the intervention as covariates. We estimate the probit model using the implementation of the program in Sao Paulo, which was a place where the program focused on attending schools with lower test scores, so treatment selection is a more severe problem in this case. We also include private schools to have a larger population for the simulation study.\footnote{Simulation results are similar if we include only public schools. }  Then we exclude the treated schools and draw placebo treatments for all schools in Brazil with a treatment selection process based on the estimated probit model. We have a population of 20,363 schools for this simulation study. Based on these simulations, we find, on average, a difference of $-0.32$ points in a standardized test score when we simply compare treated and control schools under this selection process, revealing that schools that participated in this program had, on average, worse test scores relative to other schools. 

For each realization of the placebo treatment, we control the number of treated and control observations by selecting a random sample of $N_1 \in \{ 5,10,25,50 \}$ treated and $N_0 \in \{ 50,500\}$ control schools. We then estimate the nearest neighbor matching estimator with $M \in \{1,4,10\}$ using three years of pre-intervention outcomes as matching variables. We also calculate rejection rates based on the asymptotic distribution derived by AI, and based on the randomization inference tests presented in Section 4.  For each scenario, we draw 10,000 samples.

\newpage

\section{Supplemental Tables}

 \setcounter{table}{0}
\renewcommand\thetable{A.\arabic{table}}

\setcounter{figure}{0}
\renewcommand\thefigure{A.\arabic{figure}}

\begin{table}[h!]

 \begin{center}
\caption{{\bf Empirical Monte Carlo Simulation: More Covariates}} \label{Table_more_covariates}

\begin{tabular}{lcccccccc}
\hline
\hline

     & \multicolumn{2}{c}{$M=1$}   & & \multicolumn{2}{c}{$M=4$} & & \multicolumn{2}{c}{$M=10$ }     \\ \cline{2-3} \cline{5-6}  \cline{8-9} 
     & $N_0=50$ & $N_0=500$ & & $N_0=50$ & $N_0=500$   &  & $N_0=50$ & $N_0=500$ \\

     & (1) &  (2)  & & (3) &  (4) & & (5) & (6)  \\ 
     
     \hline

\multicolumn{9}{c}{\textit{Panel A: $| \mbox{average   bias} \times 100 | $}} \\

$N_1=5$ & 3.181 & 1.703 &  & 3.999 & 2.280 &  & 4.625 & 2.509 \\
$N_1=10$ & 2.822 & 1.717 &  & 3.776 & 2.201 &  & 4.889 & 2.951 \\
$N_1=25$ & 3.005 & 1.744 &  & 3.656 & 2.196 &  & 4.538 & 2.657 \\
$N_1=50$ & 2.657 & 1.657 &  & 3.476 & 2.138 &  & 4.294 & 2.644 \\

\\
\multicolumn{9}{c}{\textit{Panel B: rejection rates based on AI  }} \\
$N_1=5$ & 0.217 & 0.218 &  & 0.211 & 0.221 &  & 0.223 & 0.222 \\
$N_1=10$ & 0.150 & 0.157 &  & 0.158 & 0.158 &  & 0.190 & 0.171 \\
$N_1=25$ & 0.136 & 0.130 &  & 0.159 & 0.145 &  & 0.184 & 0.156 \\
$N_1=50$ & 0.129 & 0.126 &  & 0.157 & 0.146 &  & 0.204 & 0.177 \\

\\
\multicolumn{9}{c}{\textit{Panel D: test based on RI. sign changes }} \\
$N_1=5$ & 0.059 & 0.067 &  & 0.003 & 0.052 &  & 0.000 & 0.024 \\
$N_1=10$ & 0.101 & 0.102 &  & 0.001 & 0.100 &  & 0.000 & 0.048 \\
$N_1=25$ & 0.123 & 0.105 &  & 0.000 & 0.111 &  & 0.000 & 0.011 \\
$N_1=50$ & 0.129 & 0.116 &  & 0.000 & 0.114 &  & 0.000 & 0.000 \\

\hline

\end{tabular}
\end{center}
{\footnotesize{Note:  This table replicates the simulations presented in Table \ref{EMC_text} with the difference that we add three additional covariates that are uncorrelated with the potential outcomes. 
}
}

\end{table}

\pagebreak

\begin{table}[h!]

\footnotesize
 \begin{center}
\caption{{\bf Empirical Monte Carlo Simulation: Alternative Tests)} } \label{Table_wild}

\begin{tabular}{lcccccccc}
\hline
\hline

     & \multicolumn{2}{c}{$M=1$}   & & \multicolumn{2}{c}{$M=4$} & & \multicolumn{2}{c}{$M=10$ }     \\ \cline{2-3} \cline{5-6}  \cline{8-9} 
     & $N_0=50$ & $N_0=500$ & & $N_0=50$ & $N_0=500$   &  & $N_0=50$ & $N_0=500$ \\

     & (1) &  (2)  & & (3) &  (4) & & (5) & (6)  \\ 
     
     \hline

\multicolumn{9}{c}{\textit{Panel A: test based on RI, permutations }} \\

$N_1=5$ & 0.052 & 0.068 &  & 0.097 & 0.098 &  & 0.074 & 0.106 \\
$N_1=10$ & 0.099 & 0.102 &  & 0.086 & 0.100 &  & 0.049 & 0.096 \\
$N_1=25$ & 0.103 & 0.101 &  & 0.057 & 0.091 &  & 0.021 & 0.076 \\
$N_1=50$ & 0.107 & 0.094 &  & 0.036 & 0.084 &  & 0.014 & 0.047 \\

\\

\multicolumn{9}{c}{\textit{Panel B: wild bootstrap, estimating $\mu_0(x)$ using all observations}} \\

$N_1=5$ & 0.124 & 0.110 &  & 0.161 & 0.154 &  & 0.174 & 0.178 \\
$N_1=10$ & 0.118 & 0.109 &  & 0.139 & 0.133 &  & 0.140 & 0.144 \\
$N_1=25$ & 0.120 & 0.104 &  & 0.128 & 0.122 &  & 0.126 & 0.119 \\
$N_1=50$ & 0.124 & 0.105 &  & 0.123 & 0.110 &  & 0.133 & 0.105 \\

\\
\multicolumn{9}{c}{\textit{Panel C: wild bootstrap, estimating $\mu_0(x)$ using anly nearest neighbors }} \\
$N_1=5$ & 0.219 & 0.203 &  & 0.186 & 0.166 &  & 0.186 & 0.178 \\
$N_1=10$ & 0.173 & 0.149 &  & 0.152 & 0.132 &  & 0.148 & 0.146 \\
$N_1=25$ & 0.162 & 0.120 &  & 0.150 & 0.124 &  & 0.144 & 0.121 \\
$N_1=50$ & 0.171 & 0.116 &  & 0.149 & 0.116 &  & 0.150 & 0.109 \\

\hline

\end{tabular}
\end{center}
{\footnotesize{Note:  This table presents rejection rates for the same simulations presented in Table \ref{EMC_text} using the permutation test proposed in Appendix \ref{Section_permutation} and  the wild bootstrap procedure proposed by \cite{Otsu}. Panel B estimates $\mu_0(x)$ using linear OLS for the full sample of controls, as  done by \cite{Otsu} in their simulations. Panel C estimates  $\mu_0(x)$ using linear OLS only for the sample of nearest neighbors, which is how the method is implemented by default in Stata.
}
}

\end{table}

\newpage

\pagebreak

\begin{table}[h!]

  \begin{center}
\caption{{\bf Empirical Monte Carlo Simulation: Test Power}} \label{EMC_power}

\begin{tabular}{lcccccccc}
\hline
\hline

     & \multicolumn{2}{c}{$M=1$}   & & \multicolumn{2}{c}{$M=4$} & & \multicolumn{2}{c}{$M=10$ }     \\ \cline{2-3} \cline{5-6}  \cline{8-9} 
     & $N_0=50$ & $N_0=500$ & & $N_0=50$ & $N_0=500$   &  & $N_0=50$ & $N_0=500$ \\

     & (1) &  (2)  & & (3) &  (4) & & (5) & (6)  \\ 
     
     \hline

\\
\multicolumn{9}{c}{\textit{Panel A: rejection rates based on AI (2006)  }} \\

$N_1=5$ & 0.477 & 0.500 &  & 0.557 & 0.604 &  & 0.551 & 0.633 \\
$N_1=10$ & 0.572 & 0.613 &  & 0.657 & 0.744 &  & 0.656 & 0.784 \\
$N_1=25$ & 0.775 & 0.859 &  & 0.848 & 0.948 &  & 0.838 & 0.958 \\
$N_1=50$ & 0.884 & 0.976 &  & 0.928 & 0.998 &  & 0.930 & 0.998 \\

\\
\multicolumn{9}{c}{\textit{Panel B: test based on RI, sign changes }} \\
$N_1=5$ & 0.162 & 0.221 &  & 0.014 & 0.237 &  & 0.000 & 0.130 \\
$N_1=10$ & 0.458 & 0.513 &  & 0.023 & 0.629 &  & 0.000 & 0.475 \\
$N_1=25$ & 0.731 & 0.829 &  & 0.003 & 0.915 &  & 0.000 & 0.516 \\
$N_1=50$ & 0.864 & 0.972 &  & 0.000 & 0.985 &  & 0.000 & 0.135 \\

\hline

\end{tabular}
\end{center}
\footnotesize Note: This table presents simulation results from the empirical MC study described in Section  \ref{Empirical_MC}, when we consider a homogeneous treatment effect of $0.20$ standard deviations in the individual-level test scores.  

\end{table}

\pagebreak

\begin{table}[h!]

  \begin{center}
\caption{{\bf Empirical Monte Carlo Simulation: 5\%-Test}} \label{Table_MC5}

\begin{tabular}{lcccccccc}
\hline
\hline

     & \multicolumn{2}{c}{$M=1$}   & & \multicolumn{2}{c}{$M=4$} & & \multicolumn{2}{c}{$M=10$ }     \\ \cline{2-3} \cline{5-6}  \cline{8-9} 
     & $N_0=50$ & $N_0=500$ & & $N_0=50$ & $N_0=500$   &  & $N_0=50$ & $N_0=500$ \\

     & (1) &  (2)  & & (3) &  (4) & & (5) & (6)  \\ 
     
     \hline

\\
\multicolumn{9}{c}{\textit{Panel A: rejection rates based on AI (2006)  }} \\

$N_1=5$ & 0.139 & 0.148 &  & 0.140 & 0.145 &  & 0.133 & 0.144 \\
$N_1=10$ & 0.093 & 0.098 &  & 0.084 & 0.089 &  & 0.090 & 0.090 \\
$N_1=25$ & 0.068 & 0.065 &  & 0.067 & 0.064 &  & 0.077 & 0.071 \\
$N_1=50$ & 0.063 & 0.055 &  & 0.071 & 0.062 &  & 0.082 & 0.064 \\

\\
\multicolumn{9}{c}{\textit{Panel B: test based on RI, sign changes }} \\
$N_1=5$ & 0.009 & 0.015 &  & 0.000 & 0.012 &  & 0.000 & 0.005 \\
$N_1=10$ & 0.049 & 0.053 &  & 0.000 & 0.046 &  & 0.000 & 0.024 \\
$N_1=25$ & 0.052 & 0.052 &  & 0.000 & 0.049 &  & 0.000 & 0.023 \\
$N_1=50$ & 0.053 & 0.050 &  & 0.000 & 0.052 &  & 0.000 & 0.004 \\

\\
\multicolumn{9}{c}{\textit{Panel C: test based on RI, permutations }} \\
$N_1=5$ & 0.009 & 0.016 &  & 0.050 & 0.052 &  & 0.040 & 0.056 \\
$N_1=10$ & 0.049 & 0.053 &  & 0.045 & 0.052 &  & 0.024 & 0.051 \\
$N_1=25$ & 0.053 & 0.049 &  & 0.025 & 0.045 &  & 0.009 & 0.035 \\
$N_1=50$ & 0.054 & 0.046 &  & 0.016 & 0.040 &  & 0.007 & 0.022 \\

\hline

\end{tabular}
\end{center}
\footnotesize Note: This table presents simulation results from the empirical MC study described in Section  \ref{Empirical_MC}. We present rejection rates for 5\%-tests when the null is true for the asymptotic test based on AI, and for the approximate randomization tests based on sign changes and on permutations.  

\end{table}

\pagebreak

\begin{table}[h!]

 \begin{center}
\caption{{\bf ``Jovem de Futuro'': Summary Statistics}} \label{summary_stat}

\begin{tabular}{lccccc}
\hline
\hline

& \multicolumn{2}{c}{Rio de Janeiro} & &  \multicolumn{2}{c}{Sao Paulo} \\    \cline{2-3} \cline{5-6} 
& Exp. Treated  &  Nonexp.  Control  & & Exp. Treated  &  Nonexp.  Control    \\
& - & - & &  - & - \\
& Exp. Control &  Exp. Control & &  Exp. Control &  Exp. Control   \\
     & (1) &  (2)  & & (3)  & (4) \\ \hline

\multicolumn{6}{c}{Panel A: Before treatment} \\ 

2007 & 0.040  & -0.091  &  & 0.116*** & 0.117*** \\
 & (0.111) & (0.082) &  & (0.042) & (0.034) \\

2008 & 0.006  & -0.136** &  & 0.091** & 0.061  \\
 & (0.098) & (0.059) &  & (0.041) & (0.046) \\
 
2009 & 0.026  & -0.122  &  & 0.030  & 0.096** \\
 & (0.111) & (0.079) &  & (0.053) & (0.045) \\
 
 \\

\multicolumn{6}{c}{Panel B: After treatment} \\ 
  
2010 & -0.063  & -0.197*** &  & 0.097* & 0.070* \\
 & (0.124) & (0.073) &  & (0.057) & (0.042) \\
 
2011 & 0.065  & -0.086  &  & 0.142*** & 0.112*** \\
 & (0.101) & (0.059) &  & (0.048) & (0.039) \\
 
2012 & 0.016  & -0.121** &  & 0.129** & 0.093** \\
 & (0.102) & (0.050) &  & (0.054) & (0.041) \\

\underline{\# of Schools} \\
~ ~ Exp. Treated & \multicolumn{2}{c}{15} & & \multicolumn{2}{c}{39} \\
~ ~ Exp. Control & \multicolumn{2}{c}{15} & & \multicolumn{2}{c}{39} \\
~ ~ Nonexp. Control & \multicolumn{2}{c}{966} & & \multicolumn{2}{c}{3481} \\

 \hline

\end{tabular}
   \end{center}
   {\footnotesize{Note: Columns 1 and 3 present differences in test scores between experimental treated and control schools, calculated using a regression with strata fixed effects,  for Rio de Janeiro and  Sao Paulo respectively. Columns 2 and 4 present differences between non-experimental public  schools and experimental control schools,  for Rio de Janeiro and Sao Paulo respectively. Test scores are normalized such that students in the experimental control group have zero mean and variance one for each year. From 2009 to 2012 there are separate test scores for math, Portuguese, natural sciences, and human sciences, so we use the average of these four scores. Robust standard errors in parentheses.  * significant at 10\%; ** significant at 5\%; *** significant at 1\%.

}
}
\end{table}

\begin{table}[h!]
 \begin{center}
\caption{{\bf MC simulation: relaxing the symmetry conditions ($\mu_1(x)=0$) }} \label{Table_MC1}
\begin{tabular}{lcccccccc}
\hline
\hline

     & \multicolumn{2}{c}{$M=1$}   & & \multicolumn{2}{c}{$M=4$} & & \multicolumn{2}{c}{$M=10$ }     \\ \cline{2-3} \cline{5-6}  \cline{8-9} 
     & AI & Sign-changes &   & AI & Sign-changes  &   & AI & Sign-changes \\

     & (1) &  (2)  & & (3) &  (4) & & (5) & (6)  \\ 
     \hline
     
\multicolumn{9}{c}{Panel A: $\mu_1(x) = 0$ and $\epsilon | (X=x,W=1) \sim N(0,1)$) (symmetric)} \\
$N_1=5$ & 0.217 & 0.072 &  & 0.210 & 0.057 &  & 0.203 & 0.053 \\
$N_1=10$ & 0.152 & 0.098 &  & 0.154 & 0.101 &  & 0.152 & 0.101 \\
$N_1=25$ & 0.118 & 0.097 &  & 0.112 & 0.091 &  & 0.114 & 0.092 \\
$N_1=50$ & 0.107 & 0.094 &  & 0.108 & 0.098 &  & 0.110 & 0.102 \\
\\
\multicolumn{9}{c}{Panel B:  $\mu_1(x) = 0$ and $\epsilon | (X=x,W=1)  \sim (\chi^2_8-8)/4$ } \\
$N_1=5$ & 0.209 & 0.070 &  & 0.219 & 0.068 &  & 0.221 & 0.064 \\
$N_1=10$ & 0.155 & 0.104 &  & 0.160 & 0.105 &  & 0.162 & 0.107 \\
$N_1=25$ & 0.116 & 0.094 &  & 0.125 & 0.104 &  & 0.130 & 0.112 \\
$N_1=50$ & 0.110 & 0.097 &  & 0.112 & 0.098 &  & 0.113 & 0.103 \\
\\
\multicolumn{9}{c}{Panel C:  $\mu_1(x) = 0$ and $\epsilon | (X=x,W=1)  \sim (\chi^2_1-1)/\sqrt{2}$ } \\
$N_1=5$ & 0.222 & 0.074 &  & 0.250 & 0.091 &  & 0.282 & 0.106 \\
$N_1=10$ & 0.166 & 0.111 &  & 0.185 & 0.133 &  & 0.208 & 0.155 \\
$N_1=25$ & 0.125 & 0.104 &  & 0.145 & 0.125 &  & 0.156 & 0.132 \\
$N_1=50$ & 0.113 & 0.102 &  & 0.115 & 0.104 &  & 0.120 & 0.113 \\

\hline
\end{tabular}
\end{center}
\footnotesize Note: This table presents rejection rates for 10\%-level tests for the MC simulations discussed in Section \ref{MC_simulations}. We present rejection rates based on the asymptotic test derived by AI and based on the sign-changes test. In all simulations, $X |(W=w) \sim N(0,1)$ for $w \in \{0,1\}$, $Y(0)|(X=x,W=0) \sim N(0,1)$ for all $x \in \mathbb{R}$, and  $\mu_1(x)=0$. Each panel presents results for different distributions of $(Y(1)-\mu_1(x)|(X=x,W=1)$. The implied distribution of $\kappa_i$ is symmetric in Panel A, and becomes more asymmetric when we go to Panel C. For each cell, we run 5000 simulations.

\end{table}

\end{document}